\documentclass[article,doublecolumn]{IEEEtran}

\usepackage{fontspec}
\usepackage{amsmath}
\usepackage{amssymb}
\usepackage{color}
\usepackage{hyperref}
\usepackage[english]{babel}

\usepackage{comment}
\usepackage{bm}
\usepackage{verbatim}
\usepackage{amsthm}
\usepackage{amstext}
\usepackage{amsbsy}
\usepackage{ifpdf}
\usepackage{enumerate}
\usepackage{amscd}
\usepackage{amsfonts}
\usepackage{graphicx}
\usepackage[all]{xy}
\usepackage{verbatim}
\usepackage{mathrsfs}
\usepackage{subcaption}

\newcommand{\F}{{\mathbb{F}}}

\newcommand{\SL}{\operatorname{SL}}

\newcommand {\ignore}[1]  {}

\newcommand{\EE}{{\mathbb{E}}}

\newcommand{\diam}{{\rm diam}}

\newcommand{\eps}{\varepsilon}

\newcommand{\LL}{{\mathcal L}}

\newcommand{\df}{{\, \stackrel{\mathrm{def}}{=}\, }}

\newcommand{\vre}{\varepsilon}

\newcommand{\PSP}{P^{\mathrm{SP}}_{n,k}}

\newcommand{\R}{{\mathbb{R}}}
\newcommand{\TT}{{\mathbb{T}}}

\newcommand{\Z}{{\mathbb{Z}}}

\newcommand{\BB}{{\mathcal{B}}}
\newcommand{\SSS}{{\mathcal{S}}}

\newcommand{\AAA}{{\mathcal{A}}}
\newcommand{\Vol}{{\mathrm{vol}}}
\newcommand{\crog}{{c_{\mathrm{Rog}}}}

\newcommand{\A}{{\mathcal{A}}}
\newcommand{\CC}{{\mathcal{C}}}
\newcommand{\Gr}{{\mathrm{Gr}}}

\newcommand{\Unif}{\mathrm{Uniform}}

\usepackage{dsfont}
\newcommand{\Ind}{\mathds{1}}

\newcommand{\KK}{{\mathcal{K}}}
\newcommand{\mVLp}{{\mathcal{V}^{(p)}_L}}
\newcommand{\mVL}{{\mathcal{V}_L}}
\newcommand{\mBp}{{\mathcal{B}^{(p)}}}
\newcommand{\mBeuc}{{\mathcal{B}}}
\newcommand{\Vp}{{V_{p,n}}}
\newcommand{\GPl}{{G^{(p)}_L}}
\newcommand{\sinc}{{\mathrm{sinc}}}
\newcommand{\rpack}{{r_{\mathrm{pack}}}}
\newcommand{\rcov}{{r_{\mathrm{cov}}}}
\newcommand{\reff}{{r_{\mathrm{eff}}}}
\newcommand{\VC}{\mathcal{V}_{\mathcal{C}}}
\newcommand{\Beq}{\bar{\BB}^{n,k}}
\newcommand{\VcdfL}{g_L}
\newcommand{\VcdfLB}{\underline{g}}
\newcommand{\LBfo}{\underline{g}_{\mathrm{Jensen}}}
\newcommand{\LBco}{\underline{g}_{\mathrm{covering}}}

\newcommand{\VcdfB}{g_{\mathcal{B}}}
\newcommand{\Pesig}{P_{e,\sigma^2}}

\newtheorem{thm}{Theorem}[section]

\newtheorem{lem}[thm]{Lemma}

\newtheorem{prop}[thm]{Proposition}

\newtheorem{cor}[thm]{Corollary}

\newtheorem{remark}[thm]{Remark}

\newif\ifdraft\drafttrue
\title{The Voronoi Spherical CDF for Lattices and Linear Codes: New Bounds for Quantization and Coding}

\author{Or Ordentlich %,~\IEEEmembership{Member,~IEEE}, 
\thanks{O. Ordentlich is with the %Rachel and Selim Benin 
Hebrew University of Jerusalem, Israel (\texttt{or.ordentlich@mail.huji.ac.il}). This work was supported by ISF 1641/21 and ISF 2878/25. }
}

\begin{document}
	\date{\today}
	\maketitle
	
\begin{abstract}
For a lattice/linear code, we define the Voronoi spherical cumulative density function (CDF) as the CDF of the $\ell_2$-norm/Hamming weight of a random vector uniformly distributed over the Voronoi cell. Using the first moment method together with a simple application of Jensen's inequality, we develop lower bounds on the expected Voronoi spherical CDF of a random lattice/linear code. Our bounds are valid for any finite dimension and are quite close to a  ball-based lower bound. They immediately translate to new non-asymptotic upper bounds on the normalized second moment and the error probability of a random lattice over the additive white Gaussian noise channel, as well as new non-asymptotic upper bounds on the Hamming distortion and the error probability of a random linear code over the binary symmetric channel. In particular, we show that for most lattices in $\mathbb{R}^n$ the second moment is  greater than that of a Euclidean ball with the same covolume only by a  $\left(1+O(\frac{1}{n})\right)$ multiplicative factor. Similarly, for most linear codes in $\mathbb{F}_2^n$ the expected Hamming distortion is greater than that of a corresponding Hamming ball only by an additive universal constant.
\end{abstract}
	
\section{Introduction and main results}

This paper studies two fundamental quantities associated with lattices in $\R^n$, as well as their counterparts for linear codes in $\F_2^n$. In particular, for lattices we study the normalized second moment (NSM) and the resilience to Gaussian iid noise. For linear codes, we study the analogous quantities: the Hamming distortion (expected Hamming distance of a uniform point on $\F_2^n$ to the code), and the resilience to Bernoulli iid noise. Those quantities are instrumental for characterizing the performance of the lattice/linear code as a quantizer, as well as its usefulness for reliable transmission of information over an additive white Gaussian noise (AWGN) channel  for the lattice case, and over a binary symmetric channel (BSC) for the linear code case.

\medskip

We now briefly describe our main results:
\begin{itemize}
\item The normalized second moment (NSM) of a unit covolume lattice $L\subset \R^n$ is defined as $G_L=\frac{1}{n}\EE\|U_L\|_2^2$, where $U_L$ is uniformly distributed over the Voronoi cell of $L$ (when $\mathrm{covol}(L)\neq 1$, one further normalizes by $(\mathrm{covol}(L))^{2/n}$).
The NSM is trivially lower bounded by the second moment of the uniform distribution on a unit-volume Euclidean ball. We show that for a random lattice in $\R^n$ drawn from the natural distribution on the space of lattices (the Haar-Siegel probability distribution $\mu_n$, to be defined in the sequel), the expected NSM exceeds this lower bound only by a factor of $1+O(\frac{1}{n})$. The exact statement is in Theorem~\ref{thm:excpectedN2Mbound} and its Corollary~\ref{cor:Markov3}, that provide explicit expressions for any $n$ and an extension for the $p$th moment of a lattice is proved in Theorem~\ref{thm:excpectedNpMbound}. In particular, Theorem~\ref{thm:excpectedN2Mbound} shows the existence of a lattice with NSM close to that of the ball by a factor $1+O(\frac{1}{n})$ (this is actually true for almost all lattices, Lemma~\ref{lem:Markov2}).
Prior to this work, the tightest asymptotic upper bound on the NSM of the ``best'' lattice in $\R^n$  was $1+O(\frac{\log n}{n})$ larger than that of the ball~\cite{ZamirFeder2002}. This bound relied on upper bounding the covering density of the best lattice. Since the covering density of any lattice in $\R^n$ is known to be $\Omega(n)$~\cite{CoxeterFewRogers}, this technique is inherently limited, and cannot yield bounds with factor smaller than $1+O(\frac{\log n}{n})$ of the ball's second moment. We note that not only does Theorem~\ref{thm:excpectedN2Mbound} improve the optimal asymptotic scaling of the NSM, but it improves upon the best known upper bounds even for moderate values of $n$. In particular, for $36\leq n\leq 48$ it attains tighter upper bounds on the NSM than the best known ones, as reported in~\cite{agrell2023best}.

Furthermore, a canonical upper bound on the NSM of the ``best'' infinite constellation with a given point density was derived by Zador~\cite{zador1982asymptotic}. To date, it was not known whether or not this bound can be attained by lattices. In Theorem~\ref{thm:NSMzador} we show that for large $n$ there are lattices attaining Zador's upper bound (to within a $e^{-\Omega(n)}$ additive term). Furthermore, in Section~\ref{subsec:CSZador} we show that both Zador's upper bound for non-lattice quantizers and our new upper bound from Theorem~\ref{thm:NSMzador} for lattice quantizers approach Conway and Sloane's conjectured lower bound on the NSM  of any quantizer~\cite{conway1985lower} at rate $O\left( \frac{\log^2 n}{n^2}\right)$. 
This can be seen as an evidence that for large $n$ lattice quantizers are as good as any other quantizer, as was essentially postulated in Gersho's conjecture~\cite{gersho79}.
\item The error probability $P_{e,\sigma^2}(L)$ of a lattice $L$ at noise level $\sigma^2$ is defined as the probability that a Gaussian iid noise with zero mean and variance $\sigma^2$ falls outside of the Voronoi cell of $L$. In Theorem~\ref{thm:latticeSPUB2} we prove a novel upper bound on the expected error probability of a random lattice (drawn from $\mu_n$). This bound is numerically shown to be similar to the best known previous upper bound due to~\cite{poltyrev94AWGN,izf12}, but we were not able to compare the two bounds analytically.
\item For a linear code $\CC\subset \F_2^n$ of dimension $0\leq k\leq n$, the Voronoi region $\mathcal{V}_\CC$ is the set of all points in $\F_2^n$ that are closer under Hamming distance to $0$ than to any other codeword in $\CC$. See Section~\ref{sec:linearcodes} for the precise definition.
The Hamming distortion $D_\CC$ is defined as the expected Hamming weight of a vector uniformly distributed over the Voronoi cell. This is also the expected Hamming distortion when quantizing a uniform random vector on $\F_2^n$ to its nearest point in $\CC$. The Hamming distortion of any such $\CC$ is trivially lower bounded by the expected Hamming weight of the uniform distribution on a quasi Hamming ball of size $2^{n-k}$ (see Lemma~\ref{lem:sphereDistortion}). In Theorem~\ref{thm:DcConstantGap} we show that for a random linear code of dimension $k$, the expected Hamming distortion exceeds this lower bound only by an additive universal constant (independent of $n$). To the best of our knowledge, such a tight characterization was not known for linear codes prior to this work.
\item The error probability $P_e(\CC,p)$ of a linear code $\mathcal{C}\subset\F_2^n$ of dimension $0\leq k\leq n$ at noise level $p$ is defined as the probability that a $\mathrm{Bernoulli}(p)$ iid noise falls outside of the Voronoi cell. Equivalently, this is the block error probability of the $[n,k]$ linear code $\CC$ over the BSC with crossover probability $p$. In Theorem~\ref{thm:PeRCUB} we prove a new upper bound on the expected error probability of a random linear code. Near capacity, this bound is numerically seen to improve upon the best known finite-blocklength error probability upper bounds for the BSC due to Poltyrev~\cite{poltyrev94} and Polyanskiy-Poor-Verd\'{u}~\cite{ppv10}. 
\end{itemize}

\subsection{Technical innovation.} We give a high-level overview of the main ideas used in the proofs for the results on lattices. The same ideas are used for the analysis of linear codes.

The Voronoi cell $\mVL$ of a lattice $L\subset\R^n$ is the set of all points in $\R^n$ that are closer to $0$ than to any other lattice point in $L$. Many of the most important figures of merit of $L$ are defined through its Voronoi cell. For instance, the packing radius is the radius of the largest ball contained in $\mVL$, the covering radius is the maximum norm of a point in $\mVL$, the second moment of the lattice is the expected energy of a random vector uniformly distributed on $\mVL$, and the lattice error probability is the probability that iid Gaussian noise falls outside of $\mVL$. Thus, studying the geometry of $\mVL$ is key to analyzing the above quantities. However, the Voronoi cell is a polytope dictated by an exponential number of lattice points, and therefore characterizing it exactly becomes intractable as $n$ increases. 

Our first observation is that many of the important lattice figures of merit, including the four mentioned above, are rotation invariant, and therefore only depend on $\mVL$ through the function
\begin{align*}
\VcdfL(r)=\frac{\left|r\mBeuc\cap \mVL\right|}{|\mVL|}, ~~~r>0,
\end{align*}
which we refer to as the Voronoi spherical cumulative density function (CDF) of the lattice. Here, $\mBeuc$ is a unit radius Euclidean ball, and $|\cdot|$ denotes the volume of a set in $\R^n$ (for a discrete set $\A$, we abuse notation and use $|\A|$ to denote the number of points in $\A$). Note that if $U_L\sim\Unif(\mVL)$, then $\VcdfL(r)=\Pr(\|U_L\|_2\leq r)$ is the CDF of its $\ell_2$-norm, justifying the name Voronoi spherical CDF.
Under the criteria mentioned above, a lattice is considered ``good'' if its Voronoi cell is ``close'' to a Euclidean ball with the same volume. In terms of the Voronoi spherical CDF, this corresponds to having large $\VcdfL(r)$ for all $r>0$. We therefore seek lower bounds on $\VcdfL(r)$, which in turn, translate to upper bounds on the NSM and error probability of the lattice $L$. 

Assume without loss of generality that $L$ has unit covolume, so that $|\mVL|=1$. Define the projection of $x\in\R^n$ to $\mVL$ as $\pi_L(x)=x-Q_L(x)$, where $Q_L(x)$ maps $x$ to its nearest neighbor in the lattice $L$. The projection of the ball $r\mBeuc$ to $\mVL$ is defined as $\pi_L(r\mBeuc)=\{\pi_L(x)~:~x\in r\mBeuc\}$. Since $\|\pi_L(x)\|\leq \|x\|$, we have that $r\mBeuc\cap \mVL=\pi_L(r\mBeuc)$. Thus, computation of $\VcdfL(r)=\left|r\mBeuc\cap \mVL \right|$ reduces to computation of $|\pi_L(r\mBeuc)|$.

In~\cite[Proof of Proposition 3.6]{ORW23}, it was observed that
\begin{align*}
|\pi_L(r\mBeuc)|&=\int_{x\in r\mBeuc}\frac{1}{|(x+L)\cap r\mBeuc|}dx\\
&=\int_{x\in r\mBeuc}\frac{1}{1+|(L\setminus\{0\})\cap (r\mBeuc-x)|}dx.
\end{align*}    
Evaluation of the expression above for a particular lattice $L$ still seems challenging. However, for a random lattice $L$, we can use the convexity of $t\mapsto\frac{1}{1+t}$ together with Jensen's inequality and obtain
\begin{align*}
\EE_L[\VcdfL(r)]\geq \int_{x\in r\mBeuc}\frac{1}{1+\EE\left[|(L\setminus\{0\})\cap (r\mBeuc-x)|\right]}dx.    
\end{align*}
By Siegel's summation formula (Minkowski-Hlawka Theorem)~\cite{SiegelFormula}, we have that if $L$ is drawn from the natural probability distribution $\mu_n$ on the space of unit covolume lattices in $\R^n$, then
\begin{align*}
\EE\left[|(L\setminus\{0\})\cap (r\mBeuc-x)|\right]=|r\mBeuc|,~~~\forall x\in r\mBeuc.
\end{align*}
Consequently,
\begin{align*}
\EE_{\mu_n}[\VcdfL(r)]\geq\frac{|r\mBeuc|}{1+|r\mBeuc|}.    
\end{align*}
Note that the ``best'' (highest) Voronoi spherical CDF we could hope for is $\min\{|r\mBeuc|,1\}$, corresponding a perfect unit-volume ball. Our lower bound on $\EE_{\mu_n}[\VcdfL(r)]$ is quite close to this utopian behavior, and consequently, using it for controlling the NSM and error probability of a random lattice yields tight bounds. See Figure~\ref{fig:VorCDF} for an illustration.

A remarkable feature of our analysis is that it relies solely on the first moment method, completely circumventing the need to deal with the intricate statistical dependencies between $k$-tuples of points ($k>2$) of a random lattice.

\subsection{Related work}

Our analysis of the NSM of a random lattice, as well as its error probability, involves studying the distribution of the distance between $X\sim\Unif([-a,a)^n)$ for large $a>0$, and a lattice $L\subset\R^n$, as characterized by the Voronoi spherical CDF $\VcdfL(r)$. Hamprecht and Agrell~\cite{hamprecht2003exploring} have previously considered the complementary function $1-\VcdfL(r)$ and called it the \emph{radial function}. They have also observed that it completely characterizes the packing and covering radii of a lattice, as well as its second moment. They further evaluated and plotted this function for $\Z^n,A_2,A^*_4,D_4,A^*_8,E_8,A^*_{16}$ and $\Lambda_{16}$. For non-lattice point configurations $L\subset\R^n$, Torquato and co-authors~\cite{torquato1990nearest,torquato2010reformulation} studied the distribution of the random variable $V=\min_{y\in L}\|X-y||$. In particular, they defined the \emph{void exclusion probability} $E_V(r)=1-\Pr(V\leq r)$. When $L$ is a lattice, this specializes to the radial function $1-\VcdfL(r)$. When the constellation $L\subset \R^n$ is generated via a random process called \emph{ghost random sequential addition (RSA)}, which does not result in a lattice, Torquato used an inclusion-exclusion argument, together with the statistical properties for ghost RSA~\cite{torquato2006exactly}, to derive an upper bound on $\EE[E_V(r)]$ as a function of the point density. This bound immediately translates to an upper bound on normalized second moment of ghost RSA processes. We  also note that for the lattice case, a random variable related to $V$ was studied in~\cite{haviv2009note}, where the distance to the lattice was normalized by the covering radius, such that the normalized value is in $[0,1]$.

\medskip

The NSM of a lattice $L\subset\R^n$ is a well-studied object, as it characterizes its performance as a quantizer under mean squared error distortion (when entropy coding is further applied for describing the obtained lattice point in bits). This is true in the limit of high resolution quantization for any continuous source satisfying mild regularity conditions, and if the source is Gaussian iid, this is true for any distortion level~\cite{ramibook}. Consequently, characterizing/approximating the optimal NSM $G_n$ at any dimension $n$ is a topic that received considerable attention~\cite{fejes1959representation,cs82vor,barnes1983optimal,conway1984voronoi,conway1985lower,ConwaySloane,ZamirFeder2002,agrell1998optimization,lyu2022better,agrell2023best,ling2024rejection,agrell2024glued,pook2024parametric,agrell2025optimization}. It is well known that the NSM is lower bounded by the second moment of a uniform distribution on a unit volume Euclidean ball. Using the trivial fact that the second moment of a lattice is upper bounded by its squared (normalized) covering radius,\footnote{The second moment of any lattice is also at least $1/3$ of its squared (normalized) covering radius, and this bound is attained with equality for $L=\Z^n$. This was conjectured by~\cite{haviv2009note} and proved in~\cite[Lemme 3.2]{autissier2013lemme} and in~\cite{magazinov2020proof}.} and the well-known fact that there exist lattices in $\R^n$ whose covering radius is only $1+O(\frac{\log n}{n})$ greater than the radius of the corresponding effective ball~\cite{Rogers_again}, it was deduced in\footnote{In fact, this result is attributed to Poltyrev in~\cite{ZamirFeder2002}}~\cite{ZamirFeder2002} that $G_n$ approaches the ball lower bound at rate $1+O(\frac{\log n}{n})$. To the best of the author's knowledge, this was the best known asymptotic upper bound on $G_n$ prior to this work. The bounds presented in this work not only improve the asymptotic convergence rate of~\cite{ZamirFeder2002} from $1+O(\frac{\log n}{n})$ to $1+O(\frac{1}{n})$, but are effective already in relatively small dimensions, and improve upon the best known NSM even for $36\leq n\leq 48$.  
The optimal normalized second moment $G_n$ is also related to the quantity $b_{2,n}$ defined by Zador~\cite{zador1982asymptotic}, which quantifies the ``quantization efficiency''~\cite{gersho79,ramibook}. The quantity $b_{2,n}$ essentially measures the smallest mean squared error (MSE) distortion of any quantizer in $\R^n$ with unit point density for a source $X\sim\Unif([-a,a)^n)$ in the limit $a\to\infty$. In particular, $G_n\geq b_{2,n}$. Zador have proved an upper bound on $b_{2,n}$ that essentially follows from drawing the quantizer reconstruction points according to a Poisson point process in $\R^n$ with unit intensity~\cite[Lemma 5]{zador1982asymptotic} (see also~\cite{torquato2010reformulation} for an analysis of the quantization efficiency for the ghost RSA point process). The resulting quantizer is not a lattice, and to date it was not known if lattices can achieve this upper bound. Our Theorem~\ref{thm:NSMzador} shows that for large $n$ there are indeed lattices that achieve Zador's bound (up to an $e^{-\Omega(n)}$ additive term). In fact, a well-known conjecture due to Gersho~\cite{gersho79} postulates that the quantization cells of the optimal quantizer are all congruent to some polytope (as is of course the case for lattice quantizers), and our Theorem~\ref{thm:NSMzador} is perhaps another evidence for the validity of this conjecture. 

\medskip

The application of lattices as codebooks for reliable communication over the AWGN channel has a rich history, dating back to Blake~\cite{blake1971leech}, de Buda~\cite{deBuda75}, Conway and Sloane~\cite{cs83VorCodes} and continuing with~\cite{forney1988cosetI,forney1988cosetII,Forney89II,poltyrev94AWGN,loeliger97,urbanke1998lattice,linder2002corrected,erez2004achieving,urietal,swannack2005,LMK06,izf12,ordentlich2016simple}, as well as many other works. See also, e.g.,~\cite{zamir2002nested,kochman2009joint,nazer2011compute,ordentlich2014approximate} and~\cite[Chapter 12]{ramibook} for applications of lattices for multi-user problems. In the classic communication setup, a rate-$R$ power-constrained codebook with $2^{nR}$ vectors is constructed by taking the intersection of $L\subset\R^n$ and a shaping region $\SSS\subset \R^n$ (which is ideally a Euclidean ball, or Euclidean ball-like) chosen such as to enforce the power constraint. In order to single-out the geometry of $L$, ignoring effects of shaping, Poltyrev~\cite{poltyrev94AWGN} studied the tradeoff between the point density of an infinite constellation and its error probability (in a properly defined sense). When the infinite constellation is a lattice, this corresponds to the question: what is the smallest probability, among all lattices with unit covolume, that an iid $\mathcal{N}(0,\sigma^2)$ Gaussian noise leaves the Voronoi cell? Poltyrev derived upper bounds on the expected error probability of a random lattice~\cite{poltyrev94AWGN} (see also~\cite{loeliger97}), and later an equivalent bound was derived in simpler form in~\cite{izf12}. It was shown in~\cite{izf12} that for $\sigma^2$ smaller than, but close to, the normalized squared radius of a ball with unit volume, the error probability is greater than that of a ball only by a constant. See~\cite[Chapter 13]{ramibook}, as well as~\cite{IZ12}, for a comprehensive analysis of the error probability of a random lattice. In Theorem~\ref{thm:latticeSPUB2} we prove a different upper bound on the expected error probability of a random lattice. This bound is close to the bound given in~\cite{izf12}, but it seems (numerically) that one bound does not dominate the other in general. We also note in passing that for a non-lattice unit density infinite constellation, tighter upper bounds on the error probability are derived in~\cite{anantharam2015capacity}.

\medskip

The Hamming distortion of linear codes in $\F_2^n$ was first studied by Goblick~\cite{goblick1963coding}, who proved that such codes can asymptotically attain Shannon's rate-distortion function for a symmetric binary source under Hamming distortion. In fact, Goblick's proof relied on showing that for every $0\leq r\leq n$ there exists of a linear code $\CC\subset\F_2^n$ for which the Voronoi spherical CDF $Q_\CC(r)$, as defined in Section~\ref{sec:linearcodes}, is large. This follows by choosing the generators of the code sequentially, and ensuring that every new generator that is added sufficiently increases the number of points in $\F_2^n$ that are $r$-covered. This type of argument is also used for showing the existence of linear codes with small covering radius~\cite{cohen1997covering}. Our bounding technique, on the other hand, enables to lower bound $\EE_\CC[Q_{\CC}(r)]$ \emph{simultaneously} for all $0\leq r\leq n$, where the expectation is with respect to the natural random linear code ensemble. It does not involve a sequential selection of the generators, and results in a short and simple derivation of an upper bound on the expected Hamming distortion of a linear code, that is greater than the ball lower bound only by an additive constant. To the best of our knowledge, such a constant gap result was not previously known for \emph{linear codes}. We also note that the problem of designing practical quantizers based on linear codes whose expected Hamming distortion approaches that of the corresponding ball at rate $1+o(1)$ was studied quite intensively. See e.g.~\cite{viterbi1974trellis,matsunaga2003coding,wainwright2010lossy,korada2010polar,mazumdar2015local}. 

\medskip

Bounding the error probability of a code with $M=2^k$ codewords for transmission over $n$ uses of a binary symmetric channel with crossover probability $p$ ($\mathrm{BSC}(p)$), and in particular when the code is linear, is one of the most classic topics in coding theory, see e.g.,~\cite{gallager1968information,poltyrev94,barg2002random,ppv10} for a very partial list of references. While the vast majority of known bounds are based on some variation of the union bound (perhaps, after discarding rare events) see~\cite[Chapter 1]{sason2006performance}, the new upper bound we derive in Theorem~\ref{thm:PeRCUB} avoids the union bound, as well as the weight distribution of the code, altogether. Instead, it bounds the error probability directly through the geometry as the Voronoi region, as captured by the Voronoi spherical CDF. As we only analyze the expected error probability of a random code from the natural ensemble, and do not explore expurgation, our bound is most useful for rates close to capacity, where the sphere-packing lower bound is exponentially tight. In this regime, the best known bound prior to this work is that of Poltyrev~\cite{poltyrev94}, which is also equivalent to the RCU bound of~\cite{ppv10}. Near capacity our new bound is numerically shown to be tighter than those bounds.

\subsection{Paper structure} All results regarding lattices are developed in Section~\ref{sec:lattices}, whereas all results regarding linear codes are developed in Section~\ref{sec:linearcodes}. The two sections are self-contained, though the techniques and ideas used in the two sections are analogous. Thus, a reader interested only in lattices but not in linear codes (or vice versa), may read only Section~\ref{sec:lattices} (or only Section~\ref{sec:linearcodes}).

\subsection{Acknowledgments}
The author is grateful to Erik Agrell, Bruce Allen, Uri Erez, Bo'az Klartag, Yury Polyanskiy, Oded Regev and Barak Weiss for many valuable discussions that improved this paper. A special thanks is due to Erik Agrell for his valuable observation that our new upper bound in Theorem~\ref{thm:NSMzador} is close to the conjectured lower bound of Conway and Sloane, which led to the analysis in Subsection~\ref{subsec:CSZador}.
We also thank the anonymous reviewers and associate editor for many helpful suggestions.

\section{Lattices}
\label{sec:lattices}
For a unit covolume lattice $L\subset \mathbb{R}^n$, we define the Voronoi region as
\begin{align}
\mVL\df\left\{x\in \R^n~:~\|x\|_2\leq \|x-y\|_2,~\forall y\in L\setminus\{0\} \right\},   
\end{align}
where $\|x\|_2=\left(\sum_{i=1}^n x_i^2\right)^{1/2}$ is the $\ell_2$ norm, and ties are broken in a systematic manner, such that $\mVL$ is a fundamental cell of $L$. Let $U_L\sim\Unif(\mVL)$. The Voronoi spherical cumulative density function (CDF) of $L$ is defined as
\begin{align}
\VcdfL(r)\df \Pr(\|U_L\|_2\leq r),~~~0\leq r< \infty.    
\end{align}
Many important properties of the lattice $L$ are encoded in its Voronoi spherical CDF. Before specifying how $\VcdfL(r)$ encodes these properties, we will need some definitions.

Recall that we denote the volume of a measurable set $\AAA\subset\R^n$ by $|\AAA|\df\Vol(\AAA)$. The unit-radius closed Euclidean ball in $\R^n$ is denoted
\begin{align}
\mBeuc\df\left\{x\in \R^n~:~\|x\|_2\leq 1 \right\},    
\end{align}
and its volume is
\begin{align}
V_n\df |\mBeuc|=\frac{\pi^{n/2}}{\Gamma(1+\frac{n}{2})},
\label{eq:Vndef}
\end{align}
where $\Gamma(\cdot)$ is the Gamma function. Let
\begin{align}
\reff\df V_n^{-\frac{1}{n}},    
\end{align}
be such that $|\reff \mBeuc|=1$.
Note that
\begin{align}
\VcdfL(r)=\frac{1}{|\mVL|}\left| r\mBeuc\cap \mVL\right|=\frac{\left| r\mBeuc\cap \mVL\right|}{V_n r^n}\cdot\left(\frac{r}{\reff}\right)^n,   
\label{eq:gasIntersecVol}
\end{align}
where the last equality follows since $L$ has unit covolume.

\begin{prop}
\label{prop:fundviacdf}
Let $L\subset \R^n$ be a unit covolume lattice in $\R^n$. Then
\begin{enumerate}
    \item The packing radius of $L$ is
    \begin{align}
     \rpack(L)=\sup\left\{r>0~:~\VcdfL(r)=\left(\frac{r}{\reff} \right)^n \right\}  
    \end{align}
    \item The covering radius of $L$ is
    \begin{align}
     \rcov(L)=\inf\left\{r>0~:~\VcdfL(r)= 1 \right\}  
    \end{align}
    \item The normalized second moment (NSM) of $L$ is
    \begin{align}
     G_L\df \frac{\EE\|U_L\|_2^2}{n}=\frac{1}{n}\int_{0}^{\infty}1-\VcdfL(\sqrt{r}) dr.
    \end{align}
    \label{fund:NSM}
    \item Let $Z\sim\mathcal{N}(0,I_n)$, The $\sigma^2$-Gaussian error probability of $L$ is 
    \begin{align}
     \Pesig(L)\df\Pr(\sigma Z\notin \mVL)=\EE\left[1- \frac{\VcdfL(\sqrt{\sigma^2 W})}{\left(\frac{\sigma^2 W}{\reff^2} \right)^{\frac{n}{2}}} \right],   
    \end{align}
    where $W\sim\chi^2_{n+2}$ is a chi-squared random variable with $n+2$ degrees of freedom.
    \label{fund:GaussPe}
\end{enumerate}
\end{prop}

The first two items in the proposition above are trivial, whereas the third item has appeared in different form in~\cite[eq. 75]{torquato2010reformulation} (see also~\cite[Section III]{allen2022performance}) for analyzing performance of non-lattice infinite constellations, and re-deriving Zador's formula. For completeness, we bring the proof of Proposition~\ref{prop:fundviacdf} in Appendix~\ref{app:contCDFproperties}. The fourth item is a consequence of the following more general proposition, whose proof is also brought in Appendix~\ref{app:contCDFproperties}.

\begin{prop}
 Let $Z\sim\mathcal{N}(0,I_n)$, let $\KK\subset \R^n$ be a measurable subset of $\R^n$, and define
\begin{align}
\mu_{\sigma^2}(\KK)\df \Pr(\sigma Z\in\KK).    
\end{align}
For $U_{\KK}\sim\Unif(\KK)$, define
\begin{align}
g_\KK(r)=\Pr(\|U_\KK\|_2\leq r)=\frac{|\KK\cap r \mBeuc|}{|\KK|}.    \label{eq:genCDF_LC}
\end{align}
 Then
 \begin{align}
\mu_{\sigma^2}(\KK)=\EE\left[ \frac{g_\KK\left(\sqrt{\sigma^2 W}\right)}{\left(\frac{\sigma^2 W}{\reff(\KK)^2} \right)^{\frac{n}{2}}} \right],
 \end{align}
 where $W\sim\chi^2_{n+2}$ is a chi-squared random variable with $n+2$ degrees of freedom, and $\reff(\KK)$ is such that $|\reff(\KK)\cdot \mBeuc|=|\KK|$. 
 \label{prop:gaussianmeasureasEE}
\end{prop}

\subsection{Estimates for the expected Voronoi spherical CDF}

We say that a CDF $F(r)$ majorizes a CDF $G(r)$, and denote $G\preceq F$, if $G(r)\leq F(r)$ for all $r>0$. 
All four lattice properties above are ``improved'' if the Voronoi spherical CDF $\VcdfL$ is replaced by a majorizing CDF. In particular, from~\eqref{eq:gasIntersecVol} it is clear that for any unit covolume lattice $\VcdfL \preceq \VcdfB$ where
\begin{align}
\VcdfB(r)=\min\left\{\left(\frac{r}{\reff}\right)^n,1 \right\}
\label{eq:gBall}
\end{align}
is the CDF function for the norm of a random variable uniformly distributed over the Euclidean ball of radius $\reff$ (such that its volume is the same as that of $\mVL$). Thus, the Euclidean ball provides ``in-existence'' bounds for the $4$ quantities above. Our goal in this paper is to develop new ``existence'' bounds for $G_L$ and $\Pesig$. To that end, we find a CDF $\VcdfLB(r)$ for which
\begin{align}
\mathbb{E}[\VcdfL]\succeq \VcdfLB,    
\end{align}
where the expectation is with respect to some probability measure on $\LL_n$, the space of lattices in $\R^n$ with covolume $1$. If $\VcdfB(r)-\VcdfLB(r)$ is small for all $r>0$, we will get useful bounds for the NSM and the Gaussian error probability of a ``typical'' lattice.

\medskip

Let $\TT_L=\R^n/L$ be the quotient torus, and let $\pi_{L}:\R^n\to\TT_{L}$ be the quotient map.  Note that $\mVL$ is isomorphic to $\TT_L$. Thus, one can think of the quotient map as $\pi_L(x)=x-Q_L(x)$, where $Q_L:\R^n\to L$ maps each point in $\R^n$ to $y\in L$, such that $x\in y+\mVL$. Let $m_{L}:\TT_L\to[0,1]$ denote the Haar probability distribution on $\TT_L$ (equivalently, one can think of $m_L$ as the uniform distribution on $\mVL$). Namely, for any $\mathcal{A}\subset\TT_L$ we have $m_{L}(\mathcal{A})=|\A|$.
Observe that 
\begin{align}
\VcdfL(r)=\Pr(\|U_L\|_2\leq r)=|r\mBeuc\cap \mVL| =m_{L}(\pi_{L}(r\mBeuc)),   
\label{eq:RprojgC}
\end{align}
where the last equality holds since, when we identify $\mVL$ with $\TT_L$, it holds that $\|\pi_{L}(x)\|_2\leq \|x\|_2$ for any $x\in\R^n$. In particular, for all $x\in r\mBeuc$ we have that $\pi_{L}(x)\in r\mBeuc$. Thus, $\pi_{L}(r\mBeuc)\subset ( r\mBeuc\cap \mVL)$. On the other hand, for any point $x\in\mVL$ we have that $\pi_L(x)=x$, and consequently $\pi_{L}(r\mBeuc)\supset\pi_{L}(r\mBeuc\cap \mVL)=(r\mBeuc\cap \mVL)$. Thus, $\pi_{L}(r\mBeuc)=(r\mBeuc\cap \mVL)$. The measure $m_{L}(\pi_{L}(\KK))$ of the projection of a compact set $\KK\subset\R^n$ to the torus has been extensively studied in works ranging from the classic papers of Rogers~\cite{Rogers_bound} and Schmidt~\cite{Schmidt-admissible} in the middle of the previous century up to the recent progress in~\cite{ORW21,ORW23}. Using these works, we now develop two lower bounds on $\EE[\VcdfL(r)]$ under two different distributions on $L$, as we elaborate below.

The collection $\LL_n$ of lattices of covolume one in $\R^n$ can be identified with the quotient $\SL_n(\R)/\SL_n(\Z)$, via the map
\begin{align}
g \SL_n(\Z) \mapsto g \Z^n \ \  (g \in \SL_n(\R)).
\label{eq: mapsto}    
\end{align}
This identification endows $\LL_n$ with a natural probability
measure; namely, there is a unique $\SL_n(\R)$-invariant Borel
probability measure on $\LL_n$. Following~~\cite{ORW21}, we will refer to this measure as the
{\em Haar-Siegel measure} and denote it by
$\mu_n$. Due to its $\SL_n(\R)$-invariance, the measure $\mu_n$ is referred to as the natural measure on the space of lattices in the literature.

\begin{thm}
For $L\sim\mu_n$ we have 
\begin{align}
 \mathbb{E}[\VcdfL(r)]\geq \LBfo(r)\df\frac{\left(\frac{r}{\reff}\right)^{n}}{1+\left(\frac{r}{\reff}\right)^{n}},~~\forall r>0.
\label{eq:gLjesnenbound}
\end{align}
\label{thm:Jensen}
\end{thm}
Using the characterization~\eqref{eq:RprojgC}, the proof of Theorem~\ref{thm:Jensen} is an immediate consequence of the derivation in~\cite[Proof of Proposition 3.6]{ORW23}. This derivation relies only on the first moment method (Siegel's summation formula/the Minkowski-Hlawka-Siegel Theorem) and an application of Jensen's inequality. For completeness, we repeat the derivation from~\cite[Proof of Proposition 3.6]{ORW23} below.

\begin{proof}[Proof of Theorem~\ref{thm:Jensen}]
Observe that for any compact set $\KK\subset\R^n$ it holds that
\begin{align}
m_L(\pi_{L}(\KK))&=\int_{x\in\KK}\frac{1}{|(x+L)\cap \KK|}dx\\
&=\int_{x\in\KK}\frac{1}{1+|(L\setminus\{0\})\cap (\KK-x)|}dx.
\label{eq:Rnprojectionvolume}
\end{align}    
The function $t\mapsto \frac{1}{1+t}$ is convex in the regime $t>0$, and we can therefore write
\begin{align}
\mathbb{E}[m_L(\pi_L(\KK))]
&= \mathbb{E}\left[\int_{x\in \KK}\frac{1}{1+|(L\setminus\{0\})\cap (\KK-x)|}dx\right]\nonumber\\
&= \int_{x\in \KK}\mathbb{E}\left[\frac{1}{1+|(L\setminus\{0\})\cap (\KK-x)|}\right]dx\label{eq:fubini}\\
&\geq \int_{x\in \KK}\frac{1}{1+\mathbb{E}\left[|(L\setminus\{0\})\cap (\KK-x)|\right]}dx\label{eq:jensen}\\
&=\int_{x\in \KK}\frac{1}{1+|\KK|}dx\label{eq:siegel}\\
&=\frac{|\KK|}{1+|\KK|},\label{eq:Eml_lb}
\end{align}
where~\eqref{eq:fubini} follows from Fubini's Theorem,~\eqref{eq:jensen} from Jensen's inequality and~\eqref{eq:siegel} from Siegel's summation formula. Taking $\KK=r\mBeuc$ and applying this in conjunction with~\eqref{eq:RprojgC} we obtain
\begin{align}
\EE[\VcdfL(r)]=\EE[m_L(\pi_L(r\mBeuc))]\geq \frac{|r\mBeuc|}{1+|r\mBeuc|}= \frac{\left(\frac{r}{\reff}\right)^{n}}{1+\left(\frac{r}{\reff}\right)^{n}},
\end{align}
as claimed.
\end{proof}

\begin{figure*}[t]
\centering

\begin{subfigure}{0.49\textwidth}
    \includegraphics[width=\textwidth]{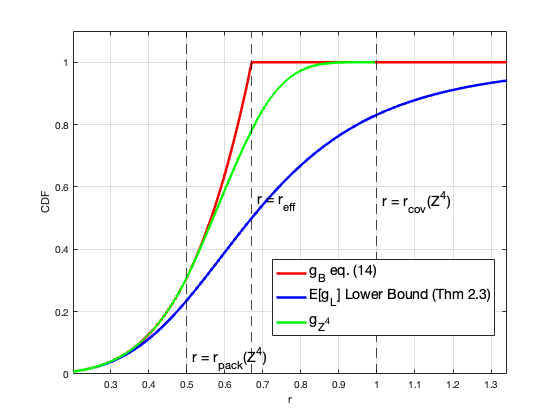}
    \caption{$n=4$}
    \label{fig:CDF4}
\end{subfigure}
\hfill
\begin{subfigure}{0.49\textwidth}
    \includegraphics[width=\textwidth]{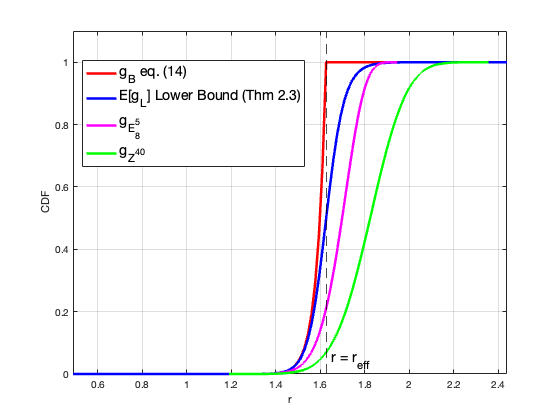}
    \caption{$n=40$}
    \label{fig:CDF40}
\end{subfigure}
\caption{Various bounds and exact expressions for the the Voronoi Spherical CDF. For dimension $n=4$ we plot $\VcdfB(r)$ from~\eqref{eq:gBall}, the lower bound on $\EE_{\mu_n}[g_L(r)]$ from~\eqref{eq:gLjesnenbound}, as well as the exact Voronoi Spherical CDF for $L=\Z^{4}$. We also indicate the points $\rpack(\Z^4)$ where $g_{\Z^4}(r)$ and $\VcdfB(r)$ diverge, $\rcov(\Z^4)$ where $g_{\Z^{4}}(r)=1$ for the first time, and $\reff$ which is common to all unit covolume lattices, and this is where $\VcdfB(r)=1$ for the first time. For dimension $n=40$. We plot $\VcdfB(r)$ from~\eqref{eq:gBall}, the lower bound on $\EE_{\mu_n}[g_L(r)]$ from~\eqref{eq:gLjesnenbound}, as well as the exact Voronoi Spherical CDF for $L=E_8^{\otimes 5}$ ($5$ copies of $E_8$) and $L=\Z^{40}$. It is seen that our bound from~\eqref{eq:gLjesnenbound} is not very effective for $n=4$, but is already quite effective for $n=40$.}
\label{fig:VorCDF}
\end{figure*}

Figure~\ref{fig:VorCDF} illustrates our bound from Theorem~\ref{thm:excpectedN2Mbound} for $n=40$, along with $\VcdfB(r)$ for the same dimension, which upper bounds $g_L(r)$ for any lattice $L\subset\R^n$. We also plot $g_L(r)$ for two particular lattices in $R^{40}$: the integer lattice $\Z^n$ and $E_8^{\otimes 5}\subset\R^{40}$ which consists of $5$ copies of the lattice $E_8$.

While $\LBfo(r)$ is quite close to $\VcdfB(r)$ for $r<\reff$, it is strictly smaller than $1$ for all $r>\reff$. This does not stem from a weakness of our bounding technique. In fact, it is a consequence of the known fact that for $L\sim\mu_n$, the probability that $L$ avoids a set $\KK\subset\R^n$ of large volume $|\KK|\gg 1$, decays only as $1/|\KK|$~\cite{Strombergtwo}. On the other hand, recent work~\cite{ORW21} showed that the covering radius of a random lattice $L\sim\mu_n$ is $\left(1+O\left(\frac{\log n}{n}\right)\right)\reff$ with probability $1-e^{-\Omega(n)}$. In light of this, we consider a distribution $\tilde{\mu}_n$ on the space of lattices $\LL_n$, obtained by conditioning $\mu_n$ on the event that $\rcov(L)$ is small.

In particular, let
\begin{align}
\eta = \eta_n \df \frac{n}{8} \log \left( \frac{4}{3} \right),
\label{eq: def eta}
\end{align}
and define the event
\begin{align}
\mathcal{E}_{\eta}=\left\{L\in\LL_n~:~\left(\frac{\rcov(L)}{\reff}\right)^n\leq 4 n^2 \eta\right\}.    
\end{align}
We define the distribution $\tilde{\mu}_n$ on $\LL_n$ as
\begin{align}
\tilde{\mu}_n\df \mu_{n|\mathcal{E}_\eta}.  
\label{eq:tildemundef}
\end{align}

Relying on classic results by Rogers~\cite{Rogers_bound} and Schmidt~\cite{Schmidt-admissible} and the more recent result~\cite{ORW21} on the covering density of a random lattice, we prove the following in Appendix~\ref{app:tildemubound}.

\begin{thm}
Assume $n\geq 13$, and define $\eta$ and $\tilde{\mu}_n$ as in~\eqref{eq: def eta} and~\eqref{eq:tildemundef}, respectively. For $L\sim\tilde{\mu}_n$ we have $\mathbb{E}[\VcdfL(r)]\geq \LBco(r)$, $\forall r>0$, where
\begin{align*}
 &\LBco(r)\nonumber\\
 &=\begin{cases}
 1-e^{-\left(\frac{r}{\reff} \right)^n}- 23 \cdot e^{-\frac{\eta}{2}} & \left(\frac{r}{\reff} \right)^n<\frac{\eta}{2}\\
 1- 24 \cdot e^{-\frac{\eta}{2}} & \frac{\eta}{2}\leq \left(\frac{r}{\reff} \right)^n< 4n^2\eta\\
 1 & \left(\frac{r}{\reff} \right)^n \geq 4n^2\eta
 \end{cases}\\
 &\geq 1- \left(e^{-\left(\frac{r}{\reff} \right)^n}+24e^{-\frac{\eta}{2}}\right)\Ind\left\{\left(\frac{r}{\reff}\right)^n<4n^2\eta\right\}
\end{align*}
\label{thm:gcovbound}
\end{thm}

We note that Theorem~\ref{thm:gcovbound} above is the only result in this paper that requires further tools beyond the first moment method.

\subsection{Bounds on the NSM}

We derive two upper bounds on $\EE[G_L]$: one for $L\sim\mu_n$ and one for $L\sim\tilde{\mu}_n$.

\begin{thm}
\label{thm:excpectedN2Mbound}
 Let $n$ be an integer and and $L\sim\mu_n$. We have that
 \begin{align}
 \EE[G_L]\leq \frac{1}{n V_n^{\frac{2}{n}}}\cdot\frac{1}{\sinc(2/n)},
 \end{align}
 where
 \begin{align}
 \sinc(t)\df\frac{\sin(\pi t)}{\pi t}.
 \end{align}
\end{thm}

Note that for $36\leq n\leq 48$ this upper bound is smaller than the best known NSM as reported in~\cite[Table 1]{agrell2023best}.  

\begin{proof}
Using Proposition~\ref{prop:fundviacdf} part~\ref{fund:NSM} (and Tonelli's Theorem) 
\begin{align}
n\EE[G_L]&=\int_{0}^{\infty} \EE[1-\VcdfL(\sqrt{r})]dr=\int_{0}^{\infty} 1-\EE[\VcdfL(\sqrt{r})]dr.
\end{align}
Applying Theorem~\ref{thm:Jensen}, we therefore have
\begin{align}
 n\EE[G_L] &\leq \int_{0}^{\infty}\frac{1}{1+\left(\frac{r}{\reff^2}\right)^{\frac{n}{2}}} dr\\
&=\reff^2\int_{0}^{\infty}\frac{1}{1+t^{\frac{n}{2}}} dt. 
\end{align}

Finally, recalling that~\cite[Section 3.241]{integraltable2007} for any $\nu>0$,
\begin{align}
\int_{0}^{\infty}\frac{1}{1+t^\nu} dt=\frac{\pi/\nu}{\sin(\pi/\nu)}=\frac{1}{\sinc(1/\nu)},
\label{eq:sincintegral}
\end{align}
and that $\reff^2=\frac{1}{V_n^{\frac{2}{n}}}$, we obtain that
\begin{align}
\EE[G_L]&\leq \frac{1}{n V_n^{\frac{2}{n}}}\cdot\frac{1}{\sinc( 2/n)}
\end{align}
and the claimed result follows.
\end{proof}

It is well-known~\cite{ConwaySloane} that among all bodies in $\R^n$ of unit volume, the second moment is minimized by $\reff\BB$, and 
\begin{align}
 G_{\reff\BB}\df \frac{\EE\|U_{\reff\BB}\|_2^2}{n}=\frac{n}{n+2}\cdot \frac{1}{n V_n^{2/n}}\df G^*_{n},
\label{eq:BallNSM}
\end{align}
where $U_{\reff\BB}\sim\Unif(\reff \BB)$. A simple application of Markov's inequality shows that for large $n$, almost all lattices have near-optimal NSM.

\begin{lem}
Assume $n\geq 8$ and let $L\sim\mu_n$. Then for any $\kappa>0$
\begin{align}
\Pr\left(G_L\geq(1+\kappa) G^*_{n} \right)\leq \frac{4}{ \kappa n}.    
\end{align}  
\label{lem:Markov2}
\end{lem}

\begin{proof}
Define the random variable $Y=Y_L=\frac{G_L}{G^*_{n}}-1$, which is non-negative with probability $1$.
By Markov's inequality, we have    
\begin{align}
\Pr\left(G_L\geq(1+\kappa) G^*_{n} \right)=\Pr(Y\geq \kappa)\leq\frac{\EE[Y]}{\kappa}.    
\end{align}
Using $\sin(x)\geq x(1-\frac{x^2}{6})$ for $x>0$, we have that for any $0<x\leq \frac{1}{\pi}$ it holds that
\begin{align}
\frac{1}{\sinc(x)}\leq\frac{\pi x}{\pi x(1-\frac{(\pi x)^2}{6})}\leq 1+\frac{1}{5}(\pi x)^2<1+2x^2.   
\end{align}
Applying Theorem~\ref{thm:excpectedN2Mbound} we obtain (for $n\geq 8$, such that $\frac{2}{n}\leq\frac{1}{\pi}$)
\begin{align}
\EE[Y]&\leq \frac{n+2}{n}\frac{1}{\sinc(2/n)}-1\leq \left(1+\frac{2}{n}\right)\left(1+\frac{8}{n^2}\right) -1\nonumber\\
&=\frac{2}{n}+\frac{8}{n^2}+\frac{16}{n^3}<\frac{4}{n},
\end{align}
which yields the claimed result.
\end{proof}

Setting $\kappa=c/n$ in Lemma~\ref{lem:Markov2}, we obtain the following straightforward corollary that shows, e.g., that there are lattices $L\subset \R^n$ with $G_L<\left(1+\frac{5}{n}\right)G_n^*$.
\begin{cor}
Assume $n\geq 8$ and let $L\sim\mu_n$. Then for any $c>0$
\begin{align}
\Pr\left(G_L\geq\left(1+\frac{c}{n}\right) G^*_{n} \right)\leq \frac{4}{c}.    
\end{align}  
\label{cor:Markov3}
\end{cor}

The normalized second moment characterizes the expected squared $\ell_2$ distance between $X$ drawn uniformly over a very large ball, and a unit covolume lattice $L$. Similarly, we can define the normalized $p$th moment $G_L^{(p)}$ of a unit covolume lattice $L$ as the expected $p$th power of the $\ell_p$ distance between $X$ and $L$. In Appendix~\ref{app:pmoment} we provide the precise definition and extend Theorem~\ref{thm:excpectedN2Mbound} for upper bound $\EE[G_L^{(p)}]$ for $L\sim\mu_n$.

Next, we derive an upper bound for the expected NSM $\EE[G_L]$ for $L\sim\tilde{\mu}_n$.

\begin{thm}
Let $n\geq 13$ be an integer, and $L\sim\tilde{\mu}_n$. We have that
\begin{align}
 \EE[G_L]\leq \frac{1}{n V_n^{\frac{2}{n}}}\left(\Gamma\left(1+\frac{2}{n}\right)+60 ne^{-\frac{\eta}{2}} \right) 
\end{align}
where $\eta=\frac{n}{8}\log\left(\frac{4}{3} \right)$ is defined in~\eqref{eq: def eta} and $\Gamma(\cdot)$ is the Gamma function.
\label{thm:NSMzador}
\end{thm}

Note that as $n$ grows, our upper bound approaches from above the well-known Zador~\cite[Lemma 5]{zador1982asymptotic} upper bound
\begin{align}
\overline{G}_n^{\mathrm{Zador}}\df   \frac{\Gamma\left(1+\frac{2}{n}\right)}{n V_n^{\frac{2}{n}}} 
\label{eq:Gzador}
\end{align}
on the smallest MSE a unit-density quantizer can achieve. In particular, Zador's upper bound shows that there exists an infinite constellation $\CC\subset\R^n$ with unit density, whose second moment is at most $\overline{G}_n^{\mathrm{Zador}}$. To the best of the author's knowledge, until this work it was not known whether or not there exist lattices in high dimensions that attain this bound. Theorem~\ref{thm:NSMzador} shows that this bound is indeed attained by lattices, up to a multiplicative factor of $1+c_1 e^{-c_2 n}$ for universal constants $c_1,c_2>0$.

\begin{proof}
Let $r_1=\reff^2 \cdot\left(\frac{\eta}{2}\right)^{\frac{2}{n}}$ and $r_2=\reff^2 \cdot\left(4n^2\eta \right)^{\frac{2}{n}}$. 
Using Proposition~\ref{prop:fundviacdf} part~\ref{fund:NSM}, together with Theorem~\ref{thm:gcovbound}, we have (using Tonelli's Theorem)
\begin{align}
n\EE[G_L]&=\int_{0}^{\infty}\EE[1-\VcdfL(\sqrt{r})]dr\\
&\leq \int_{0}^{r_1} e^{-\left(\frac{r}{\reff^2}\right)^{\frac{n}{2}}}+ 23 \cdot e^{-\eta/2}dr +\int_{r_1}^{r_2}  24 \cdot e^{-\frac{\eta}{2}}dr\\
&\leq 24 r_2 e^{-\frac{\eta}{2}}+\int_{0}^{r_1} e^{-\left(\frac{r}{\reff^2}\right)^{\frac{n}{2}}}dr.
\end{align}
Set $c=\reff^{-n}$, and make the change of variables $r=\left(\frac{u}{c}\right)^{\frac{2}{n}}$, $dr=\frac{2}{n c^{2/n}}u^{\frac{2}{n}-1}du$. The integral above is
\begin{align}
\int_{0}^{r_1} e^{-\left(\frac{r}{\reff^2}\right)^{\frac{n}{2}}}dr&=\int_{0}^{r_1} e^{-cr^{\frac{n}{2}}}dr\\
&=\frac{2}{n c^{2/n}}\int_{0}^{c r_1^{n/2}} u^{\frac{2}{n}-1}e^{-u}du\\
&\leq \reff^2\cdot \frac{2}{n}\cdot
 \int_{0}^{\infty} u^{\frac{2}{n}-1}e^{-u}du\\
&=\reff^2\cdot \frac{2}{n}\cdot \Gamma\left(\frac{2}{n}\right).
\end{align}
Recalling that $\reff^2=V_n^{-\frac{2}{n}}$ and that $\Gamma(1+\frac{2}{n})=\frac{2}{n}\Gamma(\frac{2}{n})$, we obtained
\begin{align}
\int_{0}^{r_1} e^{-\left(\frac{r}{\reff^2}\right)^{\frac{n}{2}}}dr&= V_n^{-\frac{2}{n}}\Gamma\left(1+\frac{2}{n}\right).
\end{align}

Furthermore, since $(4n^2\eta)^{2/n}<5/2$ for all $n\geq 13$, we have
\begin{align}
24 r_2\leq 24 \left(4 n^2 \eta \right)^{\frac{2}{n}}\reff^2\leq 60 \reff^2=60V_n^{-\frac{2}{n}} ,~~\forall n\geq 13.
\end{align}
This establishes our claimed result.
\end{proof}

\subsubsection{Gap to Conway and Sloane's Conjectured Bound}
\label{subsec:CSZador}

In~\cite{conway1985lower}, Conway and Sloane conjectured that the normalized second moment of any (not necessarily lattice) quantizer in $\R^n$  is lower bounded by
\begin{align}
\underline{G}^{\mathrm{CS}}_n\df     \frac{n+3-2H_{n+2}}{4n(n+1)}(n+1)^{1/n}(n!)^{4/n} f_n(n)^{2/n},
\label{eq:CSlb}
\end{align}
where $H_m=\sum_{i=1}^m\frac{1}{i}$ is a harmonic sum, and  $f_n(\cdot)$ is the Schl\"afli function. While the definition of $f_n(\cdot)$~\cite{conway1985lower} is recursive, and does not admit a closed form expression, in~\cite{rogers1958packing} (see also~\cite{rogers1961asymptotic,Shoom24}) Rogers and H. E. Daniels derived the approximation
\begin{align}\label{eq:rogers}
f_n(n)=\frac{\sqrt{n+1}}{\sqrt{2}\,e\,n!}\Bigl(\frac{2e}{\pi n}\Bigr)^{n/2}
\Bigl(1+\frac{31}{12n}+O(n^{-2})\Bigr).
\end{align}
Furthermore, in~\cite{Shoom24} an efficient  algorithm with high numerical precision was developed for computing $f_n(x)$ in the range $x\in[n-1,n+1]$. Using this algorithm, Erik Agrell~\cite{agrell2025personal} has observed numerically that Zador's upper bound~\eqref{eq:Gzador} on the NSM is remarkably close to $\underline{G}^{\mathrm{CS}}_n$ for large $n$. See Figure~\ref{fig:NSM},reproduced from~\cite{agrell2025personal} (see also~\cite[Figure 1]{allen2022performance} for a similar figure). Since our upper bound from Theorem~\ref{thm:NSMzador} is only $(1+c_1 e^{-c_2n})$ greater than Zador's upper bound, for universal $c_1,c_2>0$, the same holds also for this bound. The next lemma, proved in Appendix~\ref{app:Zador_CS}, validates the numerical findings of~\cite{agrell2025personal}.
\begin{lem}
\begin{align}
\frac{\overline{G}_n^{\mathrm{Zador}}}{\underline{G}^{\mathrm{CS}}_n}=1+O\left(\frac{\log^2(n)}{n^2}\right).   
\end{align}    
\label{lem:ZadorVsCS}
\end{lem}
To appreciate this result, note that
\begin{align}
 \frac{\overline{G}_n^{\mathrm{Zador}}}{G_n^*}=\frac{\overline{G}_n^{\mathrm{Zador}}}{\frac{1}{(n+2)V_n^{2/n}}}=1+\Omega\left(\frac{1}{n}\right).   
\end{align}    
Thus, Conway and Sloane's conjectured lower bound is significantly tighter than the ball-bound $G_n^*$ (often referred to as Zador's lower bound). As an immediate corollary of Lemma~\ref{lem:ZadorVsCS} and Theorem~\ref{thm:NSMzador} we obtain the following.
\begin{cor}
For $L\sim\tilde{\mu}_n$ we have
\begin{align}
 \frac{\EE[G_L]}{\underline{G}^{\mathrm{CS}}_n}=1+O\left(\frac{\log^2(n)}{n^2}\right).   
\end{align}    
\end{cor}
This shows that if Conway and Sloane's conjectured lower bound is correct, than there exist lattices that achieve the \emph{optimal NSM among all quantizers} to within a $1+O\left(\frac{\log^2(n)}{n^2}\right)$ factor (in fact this is true for almost all lattices). This constitutes a significant step towards establishing Gersho's conjecture, that essentially postulates that for dimension $n$ sufficiently large the best lattice quantizer is at least as good as any other quantizer in the same dimension.

\begin{figure*}[t]
\centering

\begin{subfigure}{0.49\textwidth}
    \includegraphics[width=\textwidth]{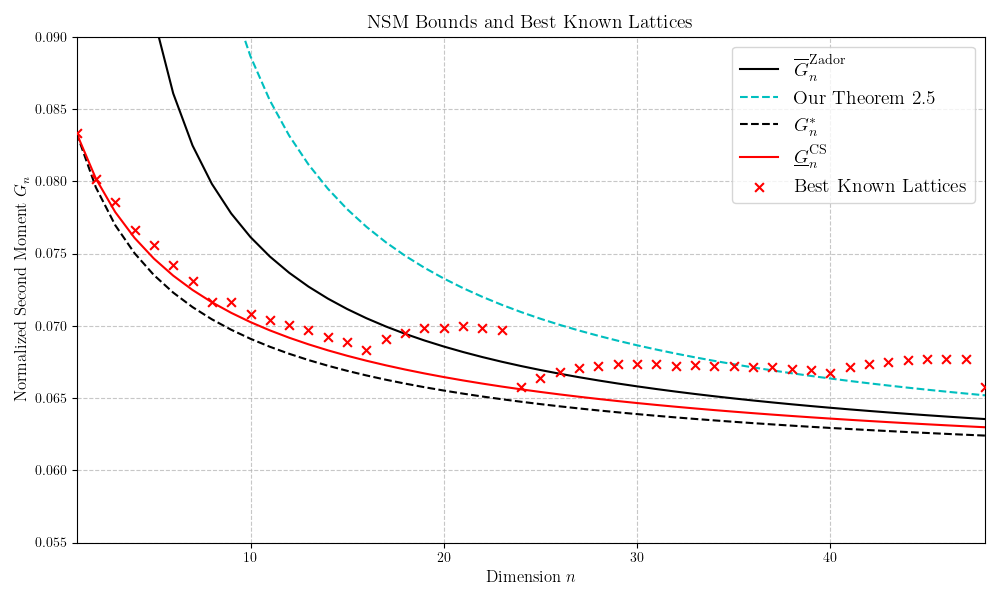}
    \caption{$n\in[1,48]$}
    \label{fig:NSMsmallN}
\end{subfigure}
\hfill
\begin{subfigure}{0.49\textwidth}
    \includegraphics[width=\textwidth]{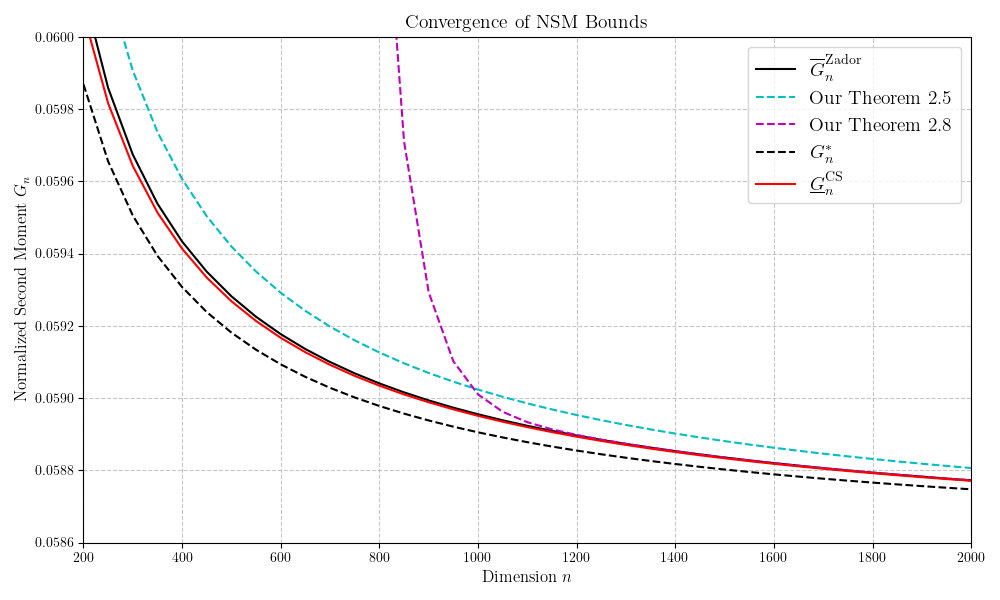}
    \caption{$n\in[200,2000]$}
    \label{fig:nsmLargeN}
\end{subfigure}
\caption{Various bounds on the optimal NSM (due to~\cite{agrell2025personal}).}\label{fig:NSM}
\end{figure*}

\subsection{Error probability of a random lattice}
The error probability of a unit covolume lattice $L\subset\R^n$ is defined as
\begin{align}
P_e(L,\sigma^2)&=\Pr(\sigma Z\notin \mVL)=1-\Pr(\sigma Z\in \mVL)\\
&=1-\mu_{\sigma^2}(\mVL).    
\end{align}
Clearly,
\begin{align}
\mu_{\sigma^2}(\mVL)\leq \mu_{\sigma^2}(\reff\mBeuc),
\end{align}
and similarly
\begin{align}
P_e(L,\sigma^2)\geq P^{\mathrm{SP}}(\sigma^2)\df1-\mu_{\sigma^2}(\reff\mBeuc).
\end{align}
We prove the following.
\begin{thm}
Let $n$ be an integer, let $W\sim\chi^2_{n+2}$ be a chi-squared random variable with $n+2$ degrees of freedom, and let
\begin{align}
X_W\df \left(\frac{\sqrt{\sigma^2 W}}{\reff}\right)^n.
\label{eq:XwDef}
\end{align}
Then, for any $\sigma^2>0$ we have that
\begin{align}
&\EE_{\mu_n}[P_e(L,\sigma^2)]\leq P^{\mathrm{SP}}(\sigma^2)+\EE\left[\frac{\min\{X_W,X^{-1}_W\}}{1+X_W} \right].
\label{eq:LatticePeUBfull2}
\end{align}
% Furthermore, for $L\sim\tilde{\mu}_n$, for $n\geq 13$, we have
% \begin{align}
% \EE_{\tilde{\mu}_n}[P_e(L,\sigma^2)]&\leq  P^{\mathrm{SP}}(\sigma^2)+\EE_W\left[\left(\frac{X_W}{2}+\frac{24e^{-\frac{\eta}{2}}}{X_W}\right)\Ind\left\{X_W<4n^2\eta\right\}\right],
% \label{eq:LatticePeUBfull3}
% \end{align}
% where $\eta$ is defined in~\eqref{eq: def eta}.
\label{thm:latticeSPUB2}
\end{thm}

\begin{proof}
By Proposition~\ref{prop:gaussianmeasureasEE}, for any unit-covolume $L\subset\R^n$ we have that
\begin{align}
\Delta&(L,\sigma^2)\df   \mu_{\sigma^2}(\reff\mBeuc)-\mu_{\sigma^2}(\mVL)\\
&=\EE\left[\frac{1}{\left(\frac{\sqrt{\sigma^2 W}}{\reff} \right)^{n}}\left(\VcdfB(\sqrt{\sigma^2 W}) -\VcdfL(\sqrt{\sigma^2 W})\right)\right]\\
&=\EE\left[\frac{1}{X_W}\left(\VcdfB\left(\reff X_W^{1/n}\right) -\VcdfL\left(\reff X_W^{1/n}\right)\right)\right].
\end{align}
Consequently,
\begin{align}
&\mathbb{E}[P_e(L,\sigma^2)]=P^{\mathrm{SP}}(\sigma^2)+\left(\mathbb{E}[P_e(L,\sigma^2)]-P^{\mathrm{SP}}(\sigma^2)\right)\\
&=P^{\mathrm{SP}}(\sigma^2)+\left(\mu_{\sigma^2}(\reff\mBeuc)-\EE[\mu_{\sigma^2}(\mVL)]\right)\\
&=P^{\mathrm{SP}}(\sigma^2)+\EE_{\mu_n}[\Delta(L,\sigma^2)]\\
&=P^{\mathrm{SP}}(\sigma^2)\nonumber\\
&+\EE_{W}\left[\frac{1}{X_W}\left(\VcdfB\left(\reff X_W^{1/n}\right) -\EE_{\mu_n}\left[\VcdfL\left(\reff X_W^{1/n}\right)\right]\right)\right]
\\
&=P^{\mathrm{SP}}(\sigma^2)\nonumber\\
&+\EE_{W}\left[\frac{1}{X_W}\left(\min\{X_W,1\} -\EE_{\mu_n}\left[\VcdfL\left(\reff X_W^{1/n}\right)\right]\right)\right]\label{eq:gballexp}\\
&\leq P^{\mathrm{SP}}(\sigma^2)+\EE_{W}\left[\frac{1}{X_W}\left(\min\{X_W,1\} -\frac{X_W}{1+X_W}\right)\right],\label{eq:usegJensPe}
\end{align}
where in~\eqref{eq:gballexp} we used $\VcdfB\left(\reff X_W^{1/n}\right)=\min\{X_W,1\}$ which follows from~\eqref{eq:gBall}, and~\eqref{eq:usegJensPe} follows from Theorem~\ref{thm:Jensen}. This establishes~\eqref{eq:LatticePeUBfull2}.
% To establish~\eqref{eq:LatticePeUBfull3}, we continue from~\eqref{eq:gballexp}, but this time taking the expectation with respect to $L\sim\tilde{\mu}_n$ (rather than $\mu_n$). Applying Theorem~\ref{thm:gcovbound}, this gives
% \begin{align}
% &\EE_{\tilde{\mu}_n}[P_e(L,\sigma^2)]\leq P^{\mathrm{SP}}(\sigma^2)\nonumber\\
% &+\EE_W\bigg(1\cdot \Ind\{X_W\leq 1\}+\frac{1}{X_W}\Ind\{X_W>1\}-\frac{1}{X_W}\\
% &+\frac{e^{-X_W}+24e^{-\frac{\eta}{2}}}{X_W}\Ind\left\{X_W<4n^2\eta\right\}\bigg)\nonumber\\
% &\leq P^{\mathrm{SP}}(\sigma^2)+\EE_W\left[\frac{e^{-X_W}-(1-X_W)+24e^{-\frac{\eta}{2}}}{X_W}\Ind\left\{X_W<4n^2\eta\right\} \right].
% \end{align}
% Since $e^{-t}-(1-t)\leq \frac{t^2}{2}$ for all $t\geq 0$, we obtain the claimed bound.
\end{proof}

One can also use Theorem~\ref{thm:gcovbound} to derive an upper bound on $\EE_{\tilde{\mu}_n}[P_e(L,\sigma^2)]$. However, such bound is not useful for most values of $n$ and $\sigma^2$ due to the $e^{-\eta/2}$ term in $\LBco(r)$.

\begin{remark}
Let $\tilde{Z}\sim\chi^2_{n}$ be a chi-squared random variable with $n$ degrees of freedom. The best known upper bound to date on $\EE[P_e(L,\sigma^2)]$,due to Poltyrev~\cite{poltyrev94AWGN} and Ingber, Zamir, and Feder~\cite[Theorem 2]{izf12}, see also~\cite[Chapter 13.4]{ramibook}, is the bound
\begin{align}
\EE_{\mu_n}[P_e(L,\sigma^2)]&\leq P_e^{\mathrm{MLB}}(\sigma^2)\df\EE\left[\min\left\{\left(\frac{\sigma^2 \tilde{Z}}{\reff^2}\right)^{\frac{n}{2}},1\right\}\right]  \nonumber\\
&=P^{\mathrm{SP}}(\sigma^2)+\EE[X_{\tilde{Z}}\Ind\{X_{\tilde{Z}}\leq 1\}],
\end{align}
where
\begin{align}
X_{\tilde{Z}}\df \left(\frac{\sqrt{\sigma^2 \tilde{Z}}}{\reff}\right)^n.
\label{eq:XzDef}
\end{align} 
The bound we obtained in Theorem~\ref{thm:latticeSPUB2}, is quite similar to $P_e^{\mathrm{MLB}}(\sigma^2)$. An analytic comparison seems non-trivial. In Figure~\ref{fig:AWGN} we evaluate the two bounds numerically, along with $P^{\mathrm{SP}}(\sigma^2)$ for $\sigma^2=\frac{0.95}{2\pi e}$ and for $\sigma^2=\frac{0.98}{2\pi e}$, and growing $n$. the two bounds are seen to be very close. 
\end{remark}

\begin{figure*}[t]
\centering

\begin{subfigure}{0.49\textwidth}
    \includegraphics[width=\textwidth]{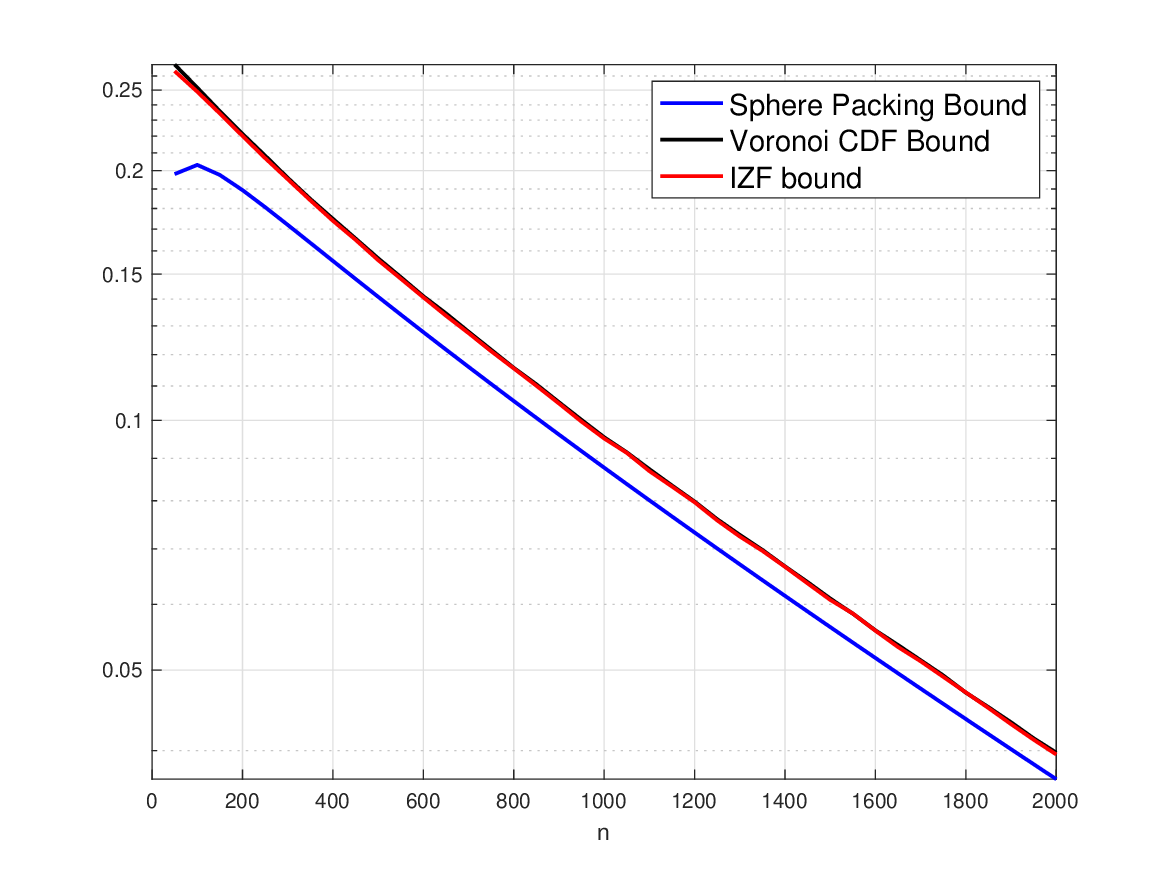}
    \caption{$\sigma^2=\frac{0.95}{2\pi e}$}
    \label{fig:sig095}
\end{subfigure}
\hfill
\begin{subfigure}{0.49\textwidth}
    \includegraphics[width=\textwidth]{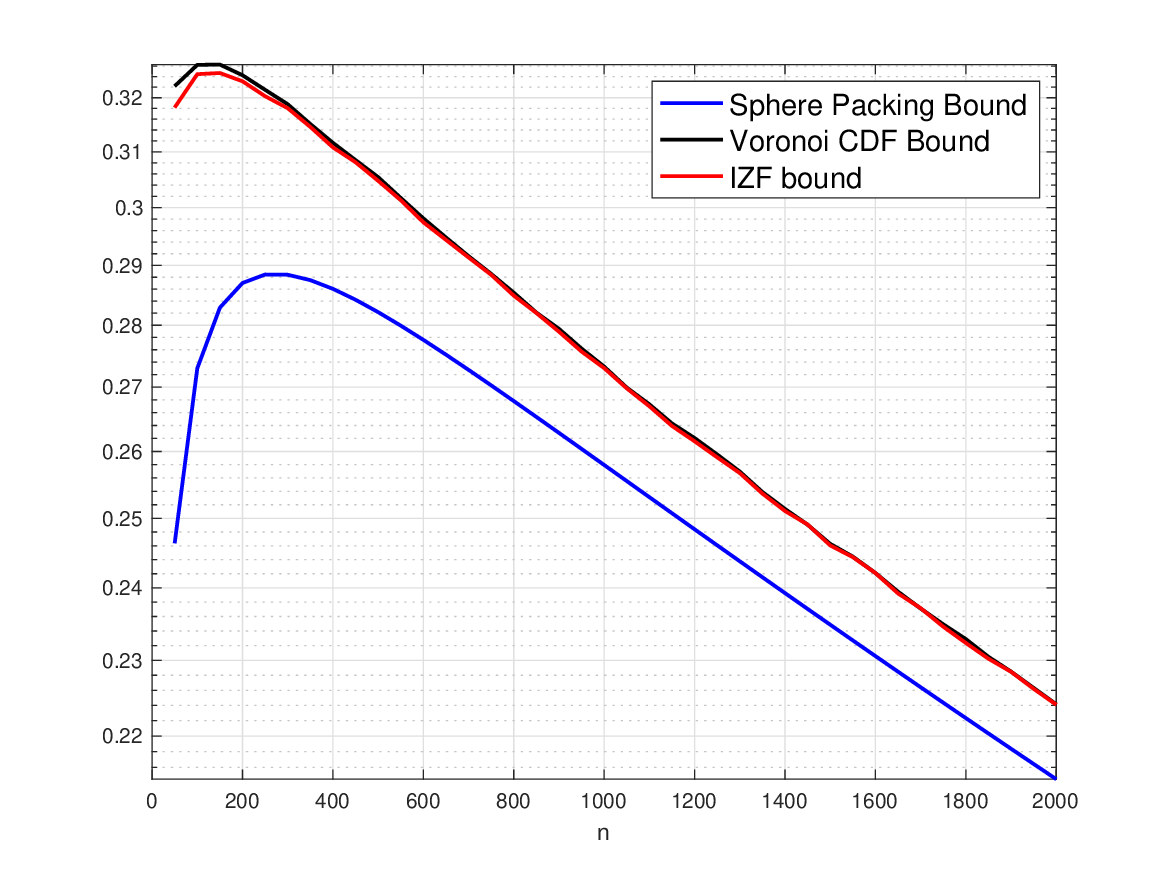}
    \caption{$\sigma^2=\frac{0.98}{2\pi e}$}
    \label{fig:sig098}
\end{subfigure}
\caption{Computation of the sphere packing lower bound, the $P_e^{\mathrm{MLB}}(\sigma^2)$ bound from~\cite[Theorem 2]{izf12}, and the new upper bound from Theorem~\ref{thm:latticeSPUB2}. The bounds are plotted for $\sigma^2=\frac{0.95}{2\pi e}$ and $\sigma^2=\frac{0.98}{2\pi e}$.}\label{fig:AWGN}
\end{figure*}

\section{Linear codes}
\label{sec:linearcodes}

Let $\CC\subset\F_2^n$ be a linear code of dimension $0\leq k\leq n$. Denote the Voronoi region of the code by
\begin{align}
\VC\df\left\{x\in\F_2^n~:~d_H(x,0)\leq d_H(x,y)~ \forall y\in\CC\setminus\{0\} \right\},
\end{align}
where $d_H(x,y)$ denotes the Hamming distance between $x,y\in\F_2^n$, and ties are broken in a systematic manner such that the cells $c+\VC$, $c\in\CC$ form a disjoint partition of $\F_2^n$. Clearly, $|\VC|=2^{n-k}$. Let $U_\CC\sim\Unif(\VC)$. The Voronoi spherical CDF of $\CC$ is defined as
\begin{align}
Q_\CC(r)\df \Pr(|U_\CC|\leq r),~~~r=0,1,\ldots,n.    
\end{align}
Here and throughout, for a vector $x\in\F_2^n$ we denote its Hamming weight by $|x|=\sum_{i=1}^n x_i$. Recall that we abuse notation and also denote the size of a set $\AAA\subset\F_2^n$ as $|\AAA|$.

As is the case for lattices, $Q_\CC(r)$ encodes many of the important properties of the code $\CC$. Before specifying how $Q_{\CC}(r)$ encodes these properties, we will need some definitions.

For integers $0\leq r\leq n$, denote the Hamming sphere of radius $r$ in $\F_2^n$ by
\begin{align}
\SSS_{n,r}\df\left\{x\in \F_2^n~:~|x|= r \right\},    
\end{align}
and the closed Hamming ball in $\F_2^n$ with radius $r$ by
\begin{align}
\BB_{n,r}\df\left\{x\in \F_2^n~:~|x|\leq r \right\}.    
\end{align}
Denote the size of $\BB_{n,r}$ as
\begin{align}
V_{n,r}=|\BB_{n,r}|=\sum_{j=0}^r {{n}\choose{j}}  .  
\end{align}
Note that
\begin{align}
Q_\CC(r)=\frac{1}{|\VC|}\left| \BB_{n,r}\cap \VC\right|=2^{k-n}\left| \BB_{n,r}\cap \VC\right|.   
\label{eq:QasIntersecVol}
\end{align}

\begin{prop}
\label{prop:LC_fundviacdf}
Let $\CC\subset \F_2^n$ be a linear code of dimension $0\leq k\leq n$. Then
\begin{enumerate}
    \item The packing radius of $\CC$ is
    \begin{align}
     \rpack(\CC)=\max\left\{r\in 0,1,\ldots,n~:~Q_\CC(r)= 2^{k-n}V_{n,r} \right\}  
    \end{align}
    \item The covering radius of $\CC$ is
    \begin{align}
     \rcov(\CC)=\min\left\{r\in 0,1,\ldots,n~:~Q_\CC(r)= 1 \right\}  
    \end{align}
    \item The Hamming distortion of $\CC$ is
    \begin{align}
     D_\CC\df \frac{\EE|U_\CC|}{n}=\frac{1}{n}\sum_{r=0}^{n} 1-Q_{\CC}(r).
    \end{align}
    \label{fund:DLN}
    \item Let $Z\sim\mathrm{Ber}^{\otimes n}(p)$, The error probability of $\CC$ for crossover probability $p$ is 
    \begin{align}
     P_e&(\CC,p)=\Pr(Z \notin \VC)\nonumber\\
     &=1- 2^{n-k}\left(\frac{p}{1-p}\cdot p^n+ \frac{1-2p}{1-p}\EE\left[\frac{Q_{\CC}(|Z|)}{{{n}\choose{|Z|}}}\right]\right), 
    \end{align}
    where $|Z|\sim\mathrm{Binomial}(n,p)$.
    \label{fund:BerPe}
\end{enumerate}
\end{prop}

The proof is given in Appendix~\ref{app:discreteCDFproperties}. The fourth item is a consequence of the following more general proposition, whose proof is also brought in Appendix~\ref{app:discreteCDFproperties}.

\begin{prop}
Let $\KK\subset \F_2^n$ be a subset, and for $U_\KK\sim\Unif(\KK)$ define
\begin{align}
Q_\KK(r)=\Pr(|U_\KK|\leq r)=\frac{|\KK\cap \BB_{n,r}|}{|\KK|}.    \label{eq:genCDF_LC}
\end{align}
Then, for $Z\sim\mathrm{Ber}^{\otimes n}(p)$ we have
\begin{align}
\mu_p(\KK)&\df \Pr(Z\in\KK)\nonumber\\
&=|\KK|\left(\frac{p}{1-p}\cdot p^n+ \frac{1-2p}{1-p}\EE\left[\frac{Q_\KK(|Z|)}{{{n}\choose{|Z|}}}\right]\right),    
\end{align}
where $|Z|\sim\mathrm{Binomial}(n,p)$.
\label{prop:ProbAsCQFexpt_dicrtete}    
\end{prop}

\subsection{Estimates for the expected Voronoi spherical CDF}

For $k
\in \{1, \ldots, n\}$, let the discrete Grassmannian
$\Gr_{n,k}(\F_2)$ denote the collection of subspaces of
dimension $k$ in $\F_2^n$, or equivalently, all linear codes of dimension $k$ in $\F_2^n$. Let $\TT_{\CC}=\F_2^n/\CC$ be the quotient group corresponding to $\CC\in\Gr_{n,k}(\F_2)$, and let $\pi_{\CC}:\F_2^n\to\TT_{\CC}$ be the projection operator to the quotient group. Note that $\VC$ is isomorphic to $\TT_\CC$. 
Thus, one can think of of the quotient map as $\pi_\CC(x)=x-f_\CC(x)$, where $f_\CC(x):\F_2^n\to\CC$ maps each point in $\F_2^n$ to its nearest codeword $c\in\CC$ such that $x\in c+\VC$. For a set $\KK\subset\F_2^n$ we denote $\pi_\CC(\KK)=\{\pi_\CC(x)~:~x\in\KK\}$. Let $m_{\CC}:\TT_\CC\to[0,1]$ denote the normalized counting measure (uniform distribution) on $\TT_\CC$. Namely, for any $\mathcal{A}\subset\TT_\CC$ we have
\begin{align}
m_{\CC}(\mathcal{A})\df\frac{|\mathcal{A}|}{|\TT_{\CC}|}=2^{k-n}|\mathcal{A}|.
\end{align}
Observe that for any $\KK\subset\F_2^n$ it holds that
\begin{align}
m_{\CC}(\pi_{\CC}(\KK))&=\sum_{x\in\KK}\frac{2^{k-n}}{|(x+\CC)\cap \KK|}\\
&=\sum_{x\in\KK}\frac{2^{k-n}}{1+|(\CC\setminus\{0\})\cap (\KK-x)|}.
\label{eq:F2nprojectionvolume}
\end{align}
Furthermore, observe that
\begin{align}
Q_\CC(r)=\Pr(|U_\CC|\leq r)=2^{k-n} |\BB_{n,r}\cap \VC| =m_{\CC}(\pi_{\CC}(\BB_{n,r})),   
\label{eq:F2mCeqQC}
\end{align}
where the last equality holds since, when we identify $\VC$ with $\TT_\CC$, it holds that $|\pi_{\CC}(x)|\leq |x|$ for any $x\in\F_2^n$. In particular, for all $x\in\BB_{n,r}$ we have that $\pi_{\CC}(x)\in\BB_{n,r}$. Thus, $\pi_{\CC}(\BB_{n,r})\subset ( \BB_{n,r}\cap \VC)$. On the other hand, for any point $x\in\VC$ we have that $\pi_\CC(x)=x$, and consequently $\pi_{\CC}(\BB_{n,r})\supset\pi_{\CC}(\BB_{n,r}\cap \VC)=(\BB_{n,r}\cap \VC)$. Thus, $\pi_{\CC}(\BB_{n,r})=(\BB_{n,r}\cap \VC)$.

\begin{thm}
Let $\CC$ be a random code drawn uniformly over $\Gr_{n,k}(\F_2)$. Then, for any $\KK\in\F_2^n$ we have that
\begin{align}
\EE[m_{\CC}(\pi_{\CC}(\KK))]\geq \frac{2^{k-n}|\KK|}{1+2^{k-n}|\KK|} . 
\label{eq:F2projectionmeasurebound}
\end{align}
In particular,
\begin{align}
\EE_\CC[Q_\CC(r)]\geq \frac{2^{k-n}V_{n,r}}{1+2^{k-n}V_{n,r}}.\label{eq:mcballRC}
\end{align}
\label{thm:RCmc}
\end{thm}

\begin{proof}
Starting from~\eqref{eq:F2nprojectionvolume}, and using Jensen's inequality with the convexity of $t\mapsto \frac{1}{1+t}$, we obtain
\begin{align}
\EE[m_{\CC}(\pi_{\CC}(\KK))]&=\EE\left[\sum_{x\in\KK}\frac{2^{k-n}}{1+|(\CC\setminus\{0\})\cap (\KK-x)|}  \right]   \\
&\geq \sum_{x\in\KK}\frac{2^{k-n}}{1+\EE[|(\CC\setminus\{0\})\cap (\KK-x)]|}.\label{eq:rondomcodebeforeexpectation}
\end{align}
Since $\CC\sim\Unif(\Gr_{n,k}(\F_2))$, for any $x\in\KK$ it holds that
\begin{align}
\EE[|(\CC\setminus\{0\})\cap (\KK-x)]|=\frac{2^k-1}{2^n-1}(|\KK|-1)<2^{k-n}|\KK|.
\end{align}
Thus,
\begin{align}
\EE[m_{\CC}(\pi_{\CC}(\KK))]
&\geq \sum_{x\in\KK}\frac{2^{k-n}}{1+2^{k-n}|\KK|}=\frac{2^{k-n}|\KK|}{1+2^{k-n}|\KK|},
\end{align}
establishing~\eqref{eq:F2projectionmeasurebound}. Combining~\eqref{eq:F2mCeqQC} with~\eqref{eq:F2projectionmeasurebound} applied with $\KK=\BB_{n,r}$, we obtain
\begin{align}
\EE[Q_\CC(r)]=\EE[\Pr(|U_\CC|\leq  r)]&\geq \frac{2^{k-n}|\BB_{n,r}|}{1+2^{k-n}|\BB_{n,r}|}\\
&=\frac{2^{k-n}V_{n,r}}{1+2^{k-n}V_{n,r}},
\end{align}
which establishes~\eqref{eq:mcballRC}.
\end{proof}

\subsection{Hamming distortion of a linear code}

As a straightforward consequence of Proposition~\ref{prop:LC_fundviacdf} and of Theorem~\ref{thm:RCmc} we obtain the following.

\begin{thm}
Let $\CC$ be a random code drawn uniformly over $\Gr_{n,k}(\F_2)$. Then,   \begin{align}
\EE[D_\CC]\leq \frac{1}{n}\sum_{r=0}^{n-1}\frac{1}{1+2^{k-n}V_{n,r}}.  
\end{align}
\label{thm:DCLN}
\end{thm}

\begin{proof}
Combining part 3 of Proposition~\ref{prop:LC_fundviacdf} and~\eqref{eq:mcballRC} of Theorem~\ref{thm:RCmc}, we obtain
\begin{align}
\EE[D_\CC]&=\frac{1}{n}\sum_{r=0}^{n-1
}\EE[1-Q_\CC(r)]\\
&\leq \frac{1}{n}\sum_{r=0}^{n-1
}1-\frac{2^{k-n}V_{n,r}}{1+2^{k-n}V_{n,r}}\\
&=\frac{1}{n}\sum_{r=0}^{n-1}\frac{1}{1+2^{k-n}V_{n,r}},
\end{align}
as claimed.
\end{proof}
    
For a set $\KK\subset\F_2^n$ of size $|\KK|$, we define $\reff=\reff(\KK)$ as the unique integer $0\leq \reff\leq n$ such that $V_{n,\reff-1}<|\KK|\leq V_{n,\reff}$, with the convention that $V_{n,-1}=-\infty$. We further define the quasi-ball $\BB_\KK$ as the union of $\BB_{n,\reff-1}$ and the first $|\KK|-V_{n,\reff-1}$ vectors of Hamming weight $\reff$ in lexicographic order (there is no significance to which $|\KK|-V_{n,\reff-1}$ vectors of Hamming weight $\reff$ we choose, we take the first ones in lexicographic order just to avoid ambiguity). For a linear code $\CC\subset\F_2^n$ of dimension $0\leq k\leq n$, we have that $\reff(\VC)$, as well as $\BB_{\VC}$, depend only on $k$ and $n$, and we therefore denote them by
\begin{align}
\reff(n,k)\df\min\{r~:~2^{n-k}\leq V_{n,r}    \},
\label{eq:reffnkdef}
\end{align}
and $\Beq$, respectively. 

Let $U_{\Beq}\sim\Unif(\Beq)$. We clearly have that for any $\CC\in\Gr_{n,k}(\F_2)$ 
\begin{align}
D_\CC\geq D^*_{n,k}=\frac{1}{n}\EE |U_{\Beq}|.    
\end{align}
See Lemma~\ref{lem:sphereDistortion} below for a more general statement.
Defining $Q_{\Beq}(r)$ as in ~\eqref{eq:genCDF_LC}, we also have that
\begin{align}
Q_{\Beq}(r)=\min\{2^{k-n}V_{n,r},1\},    
\label{eq:discreteBallCDF}
\end{align}
and therefore, using item 3 of Proposition~\ref{prop:LC_fundviacdf}, that holds for any set, not just a Voronoi region, we have
\begin{align}
D^*_{n,k}=\frac{1}{n}\sum_{r=0}^{\reff(n,k)-1} 1- 2^{k-n}V_{n,r}. 
\label{eq:DnkDef}
\end{align}
Combining this with Theorem~\ref{thm:DCLN} we obtain
\begin{align}
\Delta(n,k)&\df n(\EE[D_\CC]-D^*_{n,k})\\
&\leq\sum_{r=0}^{\reff(n,k)-1}\frac{x_r^2}{1+x_r}+\sum_{r=\reff(n,k)}^{n-1} \frac{1}{1+x_r},   
\label{eq:Deltank}
\end{align}
where
\begin{align}
x_r=x_{r,n,k}\df 2^{k-n} V_{n,r}.    
\label{eq:xrdef}
\end{align}

We prove the following result in Appendix~\ref{app:discreteCDFproperties}.

\begin{thm}
Assume $k>\alpha n$ for some $\alpha>0$. Then, there exists a universal constant $c=c_{
\alpha}>0$, independent of $n$, such that $\Delta(n,k)<c$.    
\label{thm:DcConstantGap}
\end{thm}

\begin{remark}
The result above requires that $k$ is not too small. We claim that while our requirement $k=\Omega(n)$ may possibly be relaxed, a lower bound on $k$ is unavoidable. In particular, if $k=1$, the expected distortion $D_C$ is minimized by $\CC=\{0\ldots 0,1\ldots 1\}$ (which is a perfect code, that attains the lower bound $D^*_{n,1}$). On the other hand, $\EE[D_\CC]$ is higher by $\Omega(\sqrt{n})$. 
\end{remark}

\begin{remark}
A related quantity to $D_\CC$ is the covering radius $\rcov(\CC)$. In particular, $n D_\CC\leq\rcov(\CC)$. For any code, the covering radius is at least $\reff(n,k)$, and there exist codes with covering radius $\reff(n,k)+O(\log n)$~\cite{cohen1997covering}. One can use such bounds to deduce that there exists $\CC$ with $n (D_\CC-D^*_{n,k})=O(\log n)$. Our Theorem~\ref{thm:DcConstantGap} is significantly tighter as it bounds this gap by a constant.
\end{remark}

\subsubsection{Lossy compression of binary symmetric source under Hamming distortion} Let $X\sim\mathrm{Ber}^{\otimes n}(1/2)$, and let $0\leq R\leq 1$. An $(n,R,D)$ code for $X$ consists of an encoder\footnote{For $x>1$ we write $\lfloor x \rfloor$ to denote the largest integer smaller than $x$, and $[\lfloor x \rfloor]=\{1,\ldots,\lfloor x \rfloor\}$.} $f:\F_2^n\to [\lfloor 2^{nR} \rfloor]$. and a decoder $g:[\lfloor 2^{nR} \rfloor]\to\F_2^n$, such that
\begin{align}
\frac{1}{n}\EE[d_H(X,g(f(X))]\leq D.    
\end{align}
Denote
\begin{align}
D_n(R)=\min\{D~:~\exists (n,R,D)-\text{code}\}.    
\end{align}
It is well-known~\cite{CoverThomas,PWbook} that 
\begin{align}
D(R)\leq D_n(R)< D(R)+O\left(\frac{\log n}{n}\right).
\end{align}
where
\begin{align}
D(R)=1-h_2^{-1}(R),    
\end{align}
and $h_2(p)=-p\log_2(p)-(1-p)\log_2(1-p)$, whereas $h_2^{-1}(\cdot)$ is its inverse restricted to $[0,1/2)$.
In fact,~\cite[Theorem 1]{Zhang97} shows that
\begin{align}
D_n(R)=D(R)-D'(R)\frac{\log n}{2\log n}+o\left(\frac{\log n}{n}\right) \end{align}
The following claim is essentially\footnote{Our definition of $D^*_{n,k}$ corresponds to the average distortion of a quasi-ball with size $2^{n-k}$, whereas Goblick lower bounds the expected distortion of a code with $2^k$ codewords by the average distortion of the largest ball whose size is at most than $2^{n-k}$.} stated in~\cite[eq. 4.6-4.7]{goblick1963coding}, and for completeness of exposition we provide a proof in Appendix~\ref{app:discreteCDFproperties}.
\begin{lem}
Assume $R=\frac{k}{n}$ for some integer $0\leq k\leq n$, then $D_n(R)\geq D^*_{n,k}$.
\label{lem:sphereDistortion}
\end{lem}
In light of this, and of Theorem~\ref{thm:DcConstantGap}, we have the following.
\begin{thm}
Let $0\leq k\leq n$ be an integer, and $R=\frac{k}{n}>\alpha>0$. Then
\begin{align}
D_n(R)=D^*_{n,k}+O_{\alpha}\left(\frac{1}{n}\right),    
\end{align}
and this is achievable using a linear code.
\end{thm}

\begin{proof}
The lower bound $D_n(R)\geq D^*_{n,k}$ follows from Lemma~\ref{lem:sphereDistortion}. For the upper bound, we construct $f$ and $g$ from a linear code $\CC$ as follows. We enumerate the $2^k$ points in $\CC$. The encoder $f(X)$ outputs the index of the codeword $c\in\CC$ for which $X\in c+\VC$, and the decoder $g$ outputs $c$ based on this index. The reconstruction error is uniformly distributed on $\VC$ and therefore the obtained distortion is $D_\CC$. By Theorem~\ref{thm:DcConstantGap}, there must exist a code $\CC$ with $D_\CC\leq \EE[D_\CC]<D^*_{n,k}+c_\alpha/n$, establishing our claim.
\end{proof}

\begin{remark}
Let $\CC$ be the image of the decoder $g(\cdot)$. If $\CC$ is not restricted to be a linear code (a subspace of $\F_2^n$), one can obtain even better upper bounds on the expected distortion. In particular, if we draw the $2^{nR}=2^k$ points of $\CC$ iid uniformly on $\F_2^n$, we will obtain~\cite{kostina2016nonasymptotic},\cite[Ex. V3]{PWbook}
\begin{align}
n\EE[D]&= \sum_{r=0}^n \left(1-2^{-n}V_{n,r}\right)^{2^k}\leq \sum_{r=0}^n \left(1-2^{k-n}V_{n,r}\right)\\
&\leq \sum_{r=0}^n \frac{1}{1+2^{k-n}V_{n,r}},    
\end{align}
where the last bound is the bound from Theorem~\ref{thm:DCLN}. However, the first equality relies on the statistical independence of all $2^k$ codewords. Since codewords in a linear code are just pairwise independent, this derivation is not valid for linear codes.
\end{remark}

\subsection{Error probability of a linear code}

Let $p\in(0,1/2)$ and let $Z\sim\mathrm{Ber}^{\otimes n}(p)$. For a set $\KK\subset\F_2^n$ let
\begin{align}
\mu_p(\KK)\df\Pr(Z\in\KK)=(1-p)^n\sum_{x\in\KK} \left(\frac{p}{1-p}\right)^{|x|}.
\end{align}
By definition, $|\Beq|=|\VC|$, and since $\left(\frac{p}{1-p}\right)^{|x|}$ is monotonically decreasing in $|x|$, we have that
\begin{align}
 \mu_p(\VC)\leq \mu_p(\Beq).   
\end{align}
Let $X\sim\Unif(\CC)$ be statistically independent of $Z$. The error probability in maximum-likelihood decoding of  $X$ from the output $Y=X+Z$ of the binary symmetric channel with crossover probability $p$, is 
\begin{align}
 P_e(\CC,p)=1-\mu_p(\VC).
\end{align}
Clearly,
\begin{align}
P_e(\CC,p)\geq P_e(\Beq,p)=1-\mu_p(\Beq)\df \PSP(p),    
\end{align}
which is referred to as the sphere packing lower bound in the literature. We prove the following.
\begin{thm}
Let $\CC$ be a random code drawn uniformly over $\Gr_{n,k}(\F_2)$. Then, for any $p\in(0,1/2)$ we have that
\begin{align}
&\mathbb{E}[P_e(\CC,p)]\leq \PSP(p)+\frac{1-2p}{1-p} \EE\left[\frac{V_{n,W}}{{{n}\choose{W}}}\cdot\frac{\min\{ x_{W},x^{-1}_{W}\}}{1+x_{W}} \right] \label{eq:PeAddbound1}
% \\
% &\leq\PSP(p)+\frac{1 - 2p}{1 - p} \cdot  \sum_{\ell = 0}^{n} \left( \frac{p}{1 - p} \right)^{\ell} \cdot \mathbb{E}_{W} \left[\frac{\min\{x_{W+\ell}, x_{W+\ell}^{-1}\}}{1 + x_{W+\ell}} \right],\label{eq:PeAddbound2}
\end{align}
where where $W\sim\mathrm{Binomial}(n,p)$, and $x_r=2^{k-n}V_{n,r}$.
\label{thm:PeRCUB}
\end{thm}

\begin{proof}
We have
\begin{align}
\mathbb{E}[P_e(\CC,p)]&=\PSP(p)+\left(\mathbb{E}[P_e(\CC,p)]-\PSP(p)\right)\\
&=\PSP+\left(\mu_p(\Beq)-\EE[\mu_p(\VC)]\right).\label{eq:PeDeltaEq}
\end{align}
By Proposition~\ref{prop:ProbAsCQFexpt_dicrtete} and~\eqref{eq:discreteBallCDF}, we have that
\begin{align}
&\mu_p(\Beq)\nonumber\\
&=2^{n-k}\left(\frac{p}{1-p}\cdot p^n+ \frac{1-2p}{1-p}\EE\left[\frac{\min\{2^{k-n}V_{n,|Z|},1\}}{{{n}\choose{|Z|}}}\right]\right).   
\label{eq:muBeq}
\end{align}
Similarly, by Proposition~\ref{prop:ProbAsCQFexpt_dicrtete} (or part 4 of Proposition~\ref{prop:LC_fundviacdf}) together with Theorem~\ref{thm:RCmc}, we obtain that for $\CC\sim\Unif(\Gr_{n,k}(\F_2))$
\begin{align}
&\EE[\mu_p(\VC)]\nonumber\\
&\geq 2^{n-k}\left(\frac{p}{1-p} p^n+ \frac{1-2p}{1-p}\EE\left[\frac{2^{k-n}V_{n,|Z|}}{{{n}\choose{|Z|}}\left(1+2^{k-n}V_{n,|Z|}\right)}\right]\right).       
\label{eq:muVC}
\end{align}
Consequently, we obtain
\begin{align}
&\Delta^p_{n,k}\df\mu_p(\Beq)- \EE[\mu_p(\VC)]\\
&\leq \frac{1-2p}{1-p} 2^{n-k}\EE\left[\frac{1}{{{n}\choose{|Z|}}}\left(\min\{x_{|Z|},1\}-\frac{x_{|Z|}}{1+x_{|Z|}} \right)\right]\\
&= \frac{1-2p}{1-p} 2^{n-k}\EE\left[\frac{1}{{{n}\choose{|Z|}}}\left(\frac{\min\{x^2_{|Z|},1\}}{1+x_{|Z|}} \right)\right]\\
&= \frac{1-2p}{1-p} 2^{n-k}\nonumber\\
&\cdot\EE\left[\frac{1}{{{n}\choose{|Z|}}}\left(\frac{\min\{2^{k-n}V_{n,|Z|}\cdot x_{|Z|},2^{k-n}V_{n,|Z|}\cdot x^{-1}_{|Z|}\}}{1+x_{|Z|}} \right)\right]\\
&=\frac{1-2p}{1-p} \EE\left[\frac{V_{n,|Z|}}{{{n}\choose{|Z|}}}\cdot\frac{\min\{ x_{|Z|},x^{-1}_{|Z|}\}}{1+x_{|Z|}} \right].
\label{eq:DeltaProbExact}
\end{align}
Substituting this into~\eqref{eq:PeDeltaEq} established the claim.
\end{proof}

\begin{figure*}[t]
\centering

\begin{subfigure}{0.49\textwidth}
    \includegraphics[width=\textwidth]{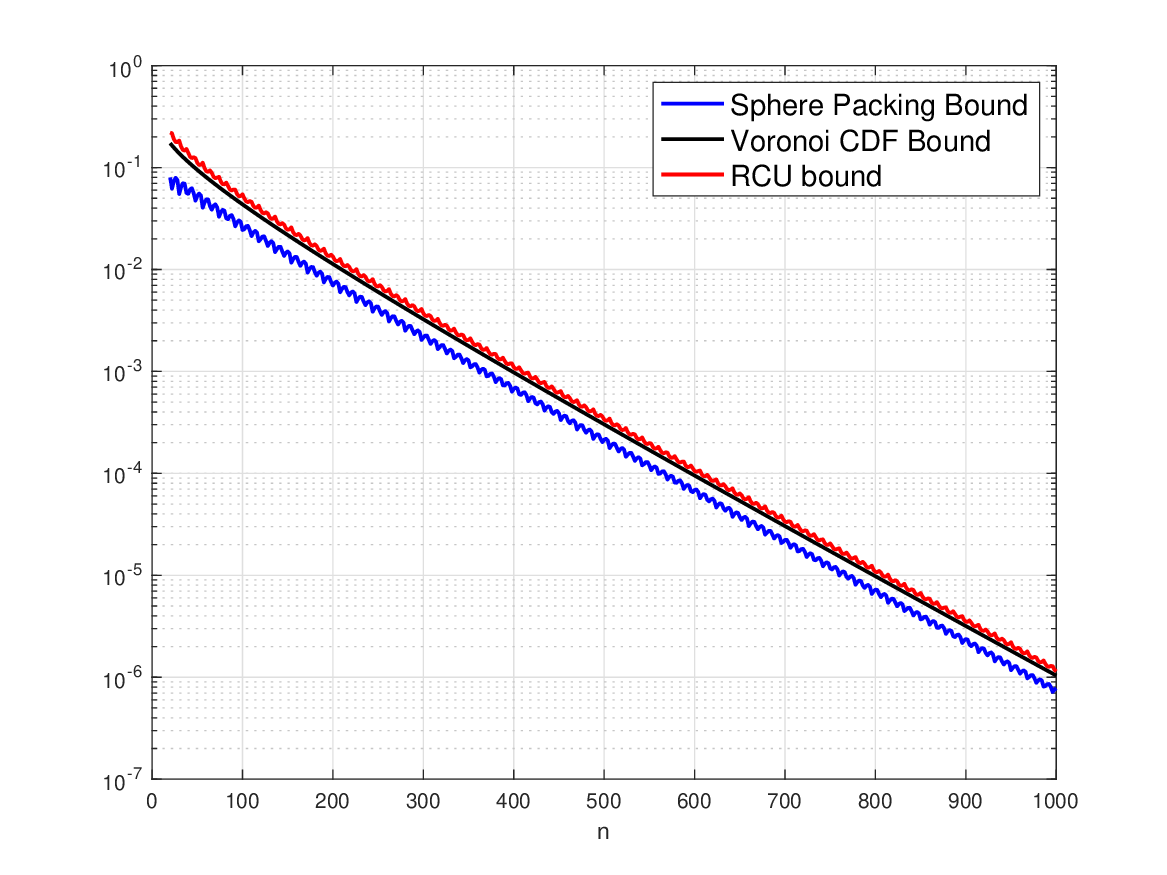}
    \caption{$p=0.07$}
    \label{fig:p007}
\end{subfigure}
\hfill
\begin{subfigure}{0.49\textwidth}
    \includegraphics[width=\textwidth]{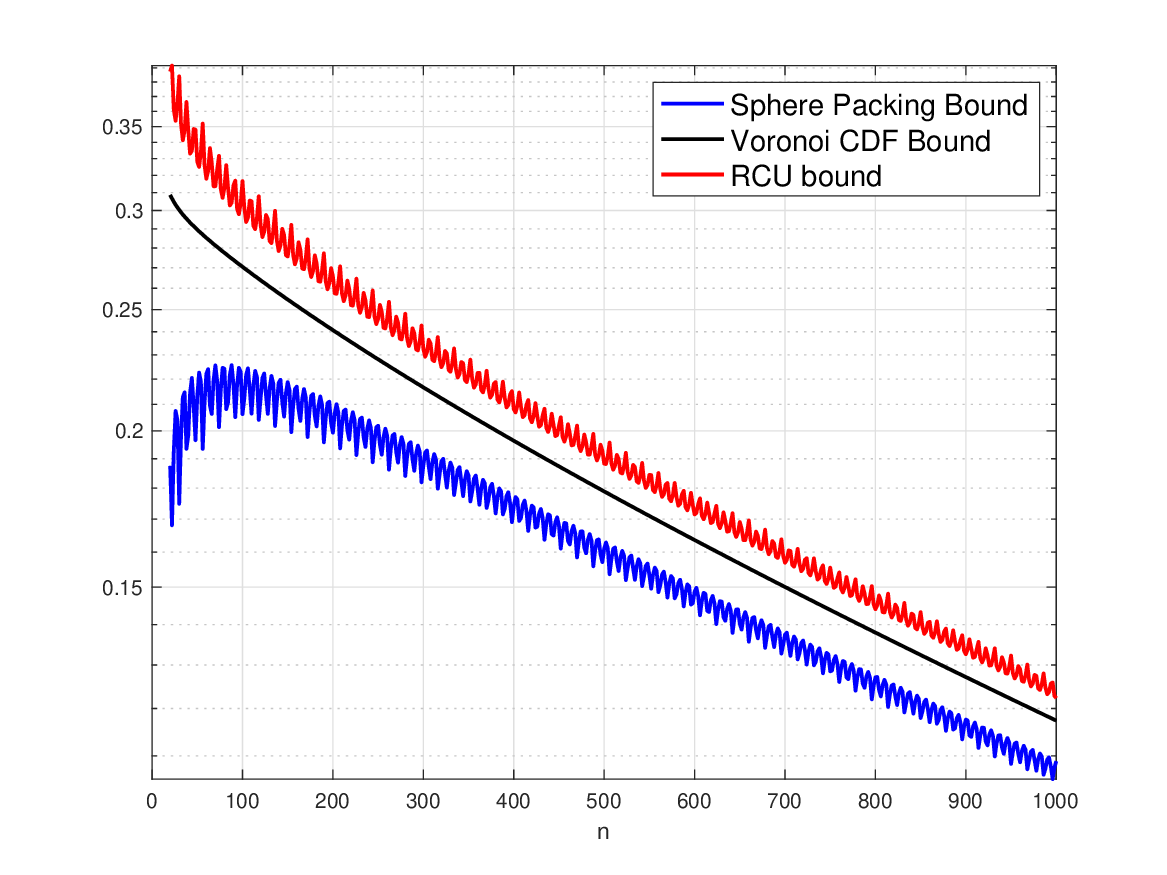}
    \caption{$p=0.1$}
    \label{fig:p01}
\end{subfigure}
\caption{Computation of the sphere packing lower bound, the RCU bound from~\cite[Theorem 33]{ppv10}, and the new upper bound from Theorem~\ref{thm:PeRCUB}. The bounds are plotted for $p=0.07$ and $p=0.1$, even values of $n$ and $k=n/2$.}\label{fig:BSC}
\end{figure*}

\begin{remark}
The benchmark upper bound to date for the binary symmetric channel is the so-called random coding upper bound (RCU bound) from~\cite{poltyrev94} and~\cite[Theorem 33]{ppv10}, which implies that for a random linear code
\begin{align}
\EE[P_e(\CC,p)]\leq \mathrm{RCU}(n,k,p)&\df\EE\left[\min\left\{1,2^{k-n}V_{n,W}\right\} \right].
\end{align}
It is not straightforward to compare the upper bound from Theorem~\ref{thm:PeRCUB} and the RCU bound. However, numerically it seems that the new bound from Theorem~\ref{thm:PeRCUB} does dominate the RCU bound. See Figure~\ref{fig:BSC} for a numeric demonstration.
\end{remark}

\appendices

\section{Proofs of claims for lattices}
\label{app:contCDFproperties}

\begin{proof}[Proof of Proposition~\ref{prop:fundviacdf}]
The first item follows since 
\begin{align}
\rpack(L)&=\sup\{r~:~r\mBeuc\subset\mVL\}\\
&=\sup\left\{r~:~\VcdfL(r)=\left(\frac{r}{\reff}\right)^n\right\},
\end{align}
where in the second equality we have used~\eqref{eq:gasIntersecVol}. The second item follows since
\begin{align}
\rcov(L)=\sup_{x\in\R^n}\min_{y\in L} \|x-y\|_2=\sup_{x\in \mVL}\|x\|_2.   
\end{align}
For the third item, we write
\begin{align}
\EE\|U_L\|^2&=\int_{r=0}^\infty \Pr(\|U_L\|^2> r)dr\\
&= \int_{r=0}^\infty 1-\Pr(\|U_L\|\leq \sqrt{r})dr.   
\end{align}
The fourth item follows directly from Proposition~\ref{prop:ProbAsCQFexpt_dicrtete}.
\end{proof}

\begin{proof}[Proof of Proposition~\ref{prop:gaussianmeasureasEE}]
Denote $C_n=C_{\sigma^2,n}=\frac{1}{(2\pi\sigma^2)^{n/2}}$ and let $U=U_{\KK}\sim\Unif(\KK)$.
It immediately follows that
\begin{align}
&\mu_{\sigma^2}(\KK)=C_n\cdot |\KK|\cdot \EE\left[e^{-\frac{\|U\|_2^2}{2\sigma^2}} \right]\\
&=C_n\cdot |\KK|\cdot \int_{0}^{\infty}\Pr\left(e^{-\frac{\|U\|_2^2}{2\sigma^2}}\geq r \right)dr\\
&=C_n\cdot |\KK|\cdot \int_{0}^{1}\Pr\left(\|U\|_2\leq\left(2\sigma^2 \log\left(\frac{1}{r}\right) \right)^{\frac{1}{2}} \right)dr\\
&=C_n\cdot V_{n}\reff^n \cdot \int_{0}^{1}\Pr\left(\|U\|_2\leq\left(2\sigma^2 \log\left(\frac{1}{r}\right) \right)^{\frac{1}{2}} \right)dr
\label{eq:muintegralformula}
\end{align}
Making the change of variables $r=e^{-\frac{w}{2}}$, $dr=-\frac{1}{2}e^{-\frac{w}{2}}dw$, we obtain
\begin{align}
 &\mu_{\sigma^2}(\KK)=C_n\cdot V_{n}\reff^n\cdot\int_{0}^{\infty} \frac{1}{2}e^{-\frac{w}{2}}\Pr\left(\|U\|_2\leq\sqrt{\sigma^2 w} \right)dw\\
 &=C_n\cdot V_{n}\reff^n \cdot 2^{\frac{n}{2}}\Gamma\left(1+\frac{n}{2}\right)\nonumber\\
&~~~~~~~\cdot\int_{0}^{\infty}\frac{w^{\frac{n+2}{2}-1}e^{-\frac{w}{2}}}{2^{\frac{n+2}{2}}\Gamma\left(\frac{n+2}{2} \right)} \frac{\Pr\left(\|U\|_2\leq\sqrt{\sigma^2 w} \right)}{w^{\frac{n}{2}}}dw\\
 &=C_n\cdot V_{n} \cdot 2^{\frac{n}{2}}\Gamma\left(1+\frac{n}{2}\right)\cdot\EE\left[\frac{\Pr\left(\|U\|_2\leq\sqrt{\sigma^2 W} \right)}{\left(\frac{W}{\reff^2}\right)^{\frac{n}{2}}} \right].
 \label{eq:chi2integralintermediate1}
\end{align}
It remains to evaluate the term multiplying the expectation. Recalling the definitions of $V_{n}$, and of $C_n$, we have
\begin{align}
C_n\cdot V_{2,n} \cdot 2^{\frac{n}{2}}\Gamma\left(1+\frac{n}{2}\right)=\frac{1}{(\sigma^2)^{\frac{n}{2}}}.  
\end{align}
Substituting this into~\eqref{eq:chi2integralintermediate1}, establishes the claim.
\end{proof}

\section{Proof of Theorem~\ref{thm:gcovbound}}
\label{app:tildemubound}
We now recall some fundamental results of Rogers and
Schmidt. For a Borel measurable subset $J \subset \R^n$, and a lattice
$L \in \LL_n$, let
$$\vre(J, L) \df 1-m_L\left(\pi_L(J) \right);$$
equivalently, 
$\vre(J, L)$ is the density of points in $\R^n$ not covered by $L
+J$. Also recall the definition of $\eta$ from~\eqref{eq: def eta}.
With these notations, the following is a corollary of \cite{Schmidt-admissible}
(see also \cite{Rogers_bound}, where a similar bound was shown for larger $\eta$, but with a large constant $\crog$ instead of the constant $7$ below):
\begin{thm}[Corollary of Theorem 4 of~\cite{Schmidt-admissible}]
\label{thm: Rogers bound}
  For all $n >13$ and
  for every Borel measurable $J \subset \R^n$
  with
  $$V \df |J| \leq \eta$$
  and bounded diameter we have
  $$\left|   \EE_{\mu_n}[\vre(J, L)] - e^{-V}
    \right | < 7 \cdot e^{-\eta}.$$  
  \end{thm}

\begin{proof}
First, it is straightforward to verify that for $V\leq \eta$ we have that
\begin{align}
e^{-V}(3/4)^{n/2}e^{4V}\leq e^{-\eta},~~ e^{-V}V^{n-1}n^{-n+1}e^{V+n}\leq e^{-\eta},    
\end{align}
and consequently for $|R|$ defined in~\cite[eq. (3)]{Schmidt-admissible}, we have $e^{-V}|R|<7e^{-\eta}$. 

Let $z=\diam(J)\cdot\sqrt{\frac{1}{4n}}\cdot[1~\cdots ~1]^\top\in\R^n$, and define $\mathrm{CUBE}(a)=z+[0,a)^n\subset\R^n$. Note that for any $t\in \mathrm{CUBE}(a)$ there is no $x\in\R^n$ such that both $x\in (t+J)$ and $-x\in (t+J)$. We have that
\begin{align}
\vre(J, L)=\lim_{a\to\infty}\frac{1}{a^n}\int_{x\in \mathrm{CUBE}(a)}\hspace{-10mm}\Ind\{|L\cap(x+J)|=\emptyset\} dx.    
\end{align}
Consequently, by the bounded convergence theorem
\begin{align}
&\EE_{\mu_n}[\vre(J, L)]\\
&=\lim_{a\to\infty}\frac{1}{a^n}\int_{x\in \mathrm{CUBE}(a)}\hspace{-10mm}\EE_{\mu_n}[\Ind\{|L\cap(x+J)|=\emptyset\}] dx\\
&=\lim_{a\to\infty}\frac{1}{a^n}\int_{x\in \mathrm{CUBE}(a)}\hspace{-12mm}\mu_n\left\{L\in\LL_n~:~L~\text{is}~(x+J)-\text{admissible} \right\} dx\\
&=\lim_{a\to\infty}\frac{1}{a^n}\int_{x\in \mathrm{CUBE}(a)}\hspace{-10mm}e^{-V}(1-R) dx\label{eq:schmidtbound}\\
&\in e^{-V}\pm 7e^{-\eta},
\end{align}
where in~\eqref{eq:schmidtbound} we have applied~\cite[Theorem 4]{Schmidt-admissible}, and in the last step we used the fact that $e^{-V}|R|\leq 7 e^{-\eta}$ provided that $|J|\leq \eta$.
\end{proof}  

Applying Theorem~\ref{thm: Rogers bound} with $J=r\BB$, whose volume is $|J|=(r/\reff)^n$, we obtain the following.
\begin{cor}
Assume $(r/\reff)^n\leq\eta$. Then, for $L\sim\mu_n$ we have
\begin{align}
\EE[g_L(r)]\geq 1-e^{-(r/\reff)^n}- 7 \cdot e^{-\eta}.   
\end{align}
\label{cor:rog}
\end{cor}

Also, a straightforward application of Markov's inequality gives the following Corollary of Theorem~\ref{thm: Rogers bound}.
\begin{cor}
For $J$ with volume $|J|=\eta$, we have
\begin{align}
\mu_n\left(\left\{L\in\LL_n~:~\vre(J,L)>e^{-\eta/2}\right\} \right)\leq 8 e^{-\eta/2}    
\end{align}
\label{cor:schmidtMarkov}
\end{cor}

Tracking down the constants in~\cite[Theorem 1.2]{ORW21}, gives the following.
\begin{thm}[Special case of Theorem 1.2 from~\cite{ORW21}] For 
\begin{align}
\mathcal{E}_\eta=\{L\in\LL_n~:~\rcov(L)\leq 4 n^2 \eta\}    
\end{align}
and $L\sim \mu_n$, we have
\begin{align}
\Pr(L\notin \mathcal{E}_\eta)\leq 16 e^{-\eta/2}.    
\end{align}
\label{thm:orwconst}
\end{thm}
To obtain this result we apply the proof of~\cite[Theorem 1.2]{ORW21}, with $V=\eta$, and $n<p<2n$, using Corollary~\ref{cor:schmidtMarkov} instead of~\cite[Corollary 2.4]{ORW21}. Note also that for $n\geq 13$ we have that $V=\eta>2\ln 2$. This gives that the covering density of $L$ is smaller or equal to $p^2 V=p^2\eta<4n^2 \eta$ with probability at least $$1-\left(8e^{-V/2}+e^{2n/p}e^{-V/2}\right)\geq 1-16 e^{-V/2}=1-16 e^{-\eta/2}.$$

\medskip

With this, we are ready to prove Theorem~\ref{thm:gcovbound}. Recall that $\mu_n$ is the Haar-Siegel probability distribution, and $\tilde{\mu}_n=\mu_{n|\mathcal{E}_\eta}$. For all $r>0$, we have
\begin{align}
&\EE_{\mu_n}[\VcdfL(r)]\nonumber\\
&=\Pr(L\in \mathcal{E}_\eta)\EE[\VcdfL(r)|L\in\mathcal{E}_\eta]+\Pr(L\notin \mathcal{E}_\eta)\EE[\VcdfL(r)|L\notin\mathcal{E}_\eta]\\
&=\Pr(L\in \mathcal{E}_\eta)\EE_{\tilde{\mu}_n}[\VcdfL(r)]+\Pr(L\notin \mathcal{E}_\eta)\EE[\VcdfL(r)|L\notin\mathcal{E}_\eta]\\
&\leq \EE_{\tilde{\mu}_n}[\VcdfL(r)]+ \Pr(L\notin \mathcal{E}_\eta).
\end{align}
Rearranging and applying Corollary~\ref{cor:rog} and Theorem~\ref{thm:orwconst}, we obtain that for $(r/\reff)^n\leq \eta/2$
\begin{align}
\EE_{\tilde{\mu}_n}[\VcdfL(r)]&\geq   1-e^{-(r/\reff)^n}- 7 \cdot e^{-\eta}-  16 e^{-\eta/2}\\
&\geq 1-e^{-(r/\reff)^n}- 23 \cdot e^{-\eta/2}.
\end{align}
By monotonicity of $r\mapsto \VcdfL(r)$, we have that $\EE_{\tilde{\mu}_n}[\VcdfL(r)]\geq 1-24 e^{-\eta/2}$ for $(r/\reff)^n\geq\eta/2$.
Finally, by definition of $\tilde{\mu}_n$ we have that for $(r/\reff)^n> 4 n^2 \eta$ it holds that $\VcdfL(r)=1$ with probability $1$ for $L\sim\tilde{\mu}_n$. This establishes the claimed result.

\section{Proof of Lemma~\ref{lem:ZadorVsCS}}
\label{app:Zador_CS}

We will show that
\begin{align}
\log \frac{\overline{G}_n^{\mathrm{Zador}}}{\underline{G}^{\mathrm{CS}}_n}=O\left(\frac{\log^2{n}}{n^2} \right),    
\end{align}
from which the claimed result immediately follows.

We use the standard asymptotic expansions for large $z$, small $\eps$ and large $n$~\cite[6.141, 6.133]{abramowitz1965handbook}~\cite[Chapter 6.3]{graham1994concrete}
\begin{align}
\log\Gamma(z)
&=\Bigl(z-\tfrac12\Bigr)\log z - z + \tfrac12\log(2\pi)
+O(z^{-1}), \label{eq:stirling}\\
\log\Gamma(1+\varepsilon)
&=-\gamma\,\varepsilon+O(\varepsilon^2),
 \label{eq:taylorGamma}\\
H_{n+2}
&=\log n+\gamma+O(n^{-1}). \label{eq:harmonic}
\end{align}
Here $\gamma$ is the Euler--Mascheroni constant.

\medskip
We start by expanding  $\log \overline{G}_n^{\mathrm{Zador}}$.
By definition of $V_n$ in~\eqref{eq:Vndef}, we have
\begin{align}
\log \overline{G}_n^{\mathrm{Zador}}= -\log(n\pi)+\frac{2}{n}\log\Gamma\!\Bigl(1+\frac n2\Bigr)
+\log\Gamma\!\Bigl(1+\frac{2}{n}\Bigr).
\end{align}
Applying \eqref{eq:stirling} with $z=1+\frac {n}{2}=\frac{n}{2}\left(1+\frac{2}{n} \right)$ and \eqref{eq:taylorGamma} with
$\varepsilon=\frac{2}{n}$ yields after some straightforward algebra
\begin{align}
\frac{2}{n}\log\Gamma\!\Bigl(1+\frac n2\Bigr)&=\left(\log\frac{n}{2}-1 \right)+\frac{\log(\pi n)}{n}+O(n^{-2}),         \\
\log\Gamma\!\Bigl(1+\frac{2}{n}\Bigr)&=\frac{-2\gamma}{n}+O(n^{-2}),
\end{align}
so that
\begin{align}\label{eq:logu}
\log \overline{G}_n^{\mathrm{Zador}}= -\log(2\pi)-1+\frac{\log(\pi n)-2\gamma}{n}
+O(n^{-2}).
\end{align}

\medskip
We proceed to expand $\log \underline{G}^{\mathrm{CS}}_n$.
Inserting \eqref{eq:rogers} into~\eqref{eq:CSlb}, and using
$\log(1+\frac{31}{12n}+O(n^{-2}))=O(n^{-1})$ gives
\begin{align}
&\log \underline{G}^{\mathrm{CS}}_n
= \log\!\Bigl(\frac{n+3-2H_{n+2}}{4n(n+1)}\Bigr)
+\frac{\log(n+1)}{n}\nonumber\\
&~~~~~~~~~~~+\frac{4}{n}\log(n!)
+\frac{2}{n}\log f_n(n)\\
&= \log\!\Bigl(\frac{n+3-2H_{n+2}}{4n(n+1)}\Bigr)
+\log\!\Bigl(\frac{2e}{\pi n}\Bigr)
+\frac{2}{n}\log(n!)\nonumber\\
&+\frac{2\log(n+1)}{n}
-\frac{\log 2+2}{n}
+O(n^{-2}).\label{eq:logell_pre}
\end{align}
% Applying \eqref{eq:harmonic}, 
% we obtain
% \begin{align}
% n+3-2H_{n+2}
% = n\Bigl(1-\frac{2\log n+2\gamma}{n}+O(n^{-1})\Bigr)
% \end{align}
% and consequently
% \begin{align}
% \log\!\bigl(n+3-2H_{n+2}\bigr)
% &= \log n -\frac{2\log n+2\gamma}{n}
% -\frac{(2\log n+2\gamma)^2}{2n^2}
% +O\!\Bigl(\frac{\log^3 n}{n^3}\Bigr)
% \nonumber\\
% &= \log n -\frac{2\log n+2\gamma}{n}
% -\frac{2\log^2 n}{n^2}
% +O\!\Bigl(\frac{\log n}{n^2}\Bigr).
% \label{eq:logHn}
% \end{align}
Applying \eqref{eq:harmonic}, 
we obtain
\begin{align}
&\frac{n+3-2H_{n+2}}{n+1}
= \frac{n+1-2(H_{n+2}-1)}{n+1}\\
&=1-\frac{2(H_{n+2}-1)}{n+1}\\
&=1-2\frac{\log(n)+(\gamma-1)+O(n^{-1})}{n+1}\\
&=1-\frac{2\log{n}+2(\gamma-1)}{n}+O\left(\frac{\log n}{n^2}\right)
\end{align}
and consequently
\begin{align}
&\log\frac{n+3-2H_{n+2}}{4n(n+1)}\nonumber\\
&= -\log(4n)\nonumber\\
&+\log\left(1-\frac{2\log{n}+2(\gamma-1)}{n}+O\left(\frac{\log n}{n^2}\right) \right)\\
&=-\log(4n)\nonumber\\
&-\frac{2\log{n}+2(\gamma-1)}{n}-\frac{2\log^2(n)}{n}+O\left(\frac{\log(n)}{n^2} \right)
\label{eq:logHn}
\end{align}
Recalling that $n!=\Gamma(1+n)$ and applying \eqref{eq:stirling} at integer $z=n+1$ (which is just Stirling's approximation) we obtain
\begin{align}
\frac{2}{n}\log(n!)=2\log(n)-2+\frac{\log(2\pi n)}{n}+O(n^{-2}).\label{eq:lognfact}   
\end{align}
Substituting~\eqref{eq:logHn} and~\eqref{eq:lognfact} into~\eqref{eq:logell_pre} gives (after careful algebra)
\begin{align}\label{eq:logell}
\log \underline{G}^{\mathrm{CS}}_n= -\log(2\pi)-1+\frac{\log(\pi n)-2\gamma}{n}
\nonumber\\-\frac{2\log^2 n}{n^2}+O\left(\frac{\log n}{n^2}\right).
\end{align}

Combining~\eqref{eq:logu} with~\eqref{eq:logell} gives
\begin{align}
\log\frac{\overline{G}_n^{\mathrm{Zador}}}{\underline{G}^{\mathrm{CS}}_n}
=\frac{2\log^2 n+O(\log n)}{n^2}=O\left(\frac{\log^2 n}{n^2} \right),
\end{align}
as claimed.

\section{Normalized $p$th Moment}
\label{app:pmoment}
Let $p>0$ be a real number. For a unit covolume lattice $L\subset \mathbb{R}^n$ and the $\ell_p$ norm $\|\cdot\|_p$ on $\R^n$, we define the $p$-Voronoi region as
\begin{align}
\mVLp\df\left\{x\in \R^n~:~\|x\|_p\leq \|x-y\|_p,~\forall y\in L\setminus\{0\} \right\},   
\end{align}
where ties are broken in a systematic manner, such that $\mVLp$ is a fundamental cell of $L$. Let 
$\TT_L \df \R^n/L$ be the quotient torus, which is isomorphic to $\mVLp$, and let $m_L$ be the Haar probability
measure on $\TT_L$, and $\pi_L : \R^n \to
\TT_L$ be the quotient map.  

Let $U^{(p)}_L=U_L\sim\Unif(\mVLp)$. We define the normalized $p$th moment of $L$ as
\begin{align}
\GPl\df \frac{\EE(\|U_L\|^p_p)}{n}
\end{align}
Let 
\begin{align}
\mBp\df\{x\in\R^n~:~\|x\|_p\leq  1\}    
\end{align}
be the unit closed ball with respect to $\|\cdot\|_p$, ans let $\Vp$ be its volume. The main result of this section, is the following extension of Theorem~\ref{thm:excpectedN2Mbound}.
\begin{thm}
\label{thm:excpectedNpMbound}
 Let $n$ be an integer and $p>0$, and $L\sim\mu_n$. We have that
 \begin{align}
 \EE[\GPl]\leq \frac{1}{n\Vp^{p/n}}\cdot\frac{1}{\sinc(p/n)}.
 \end{align}
 \end{thm}

\begin{proof}
Fix a unit covolume lattice $L\subset\R^n$. Since $\|U_L\|_p$ is a non-negative random variable, we have that
\begin{align}
\EE\|U_L\|^p_p&=\int_{0}^{\infty}\Pr(\|U_L\|^p_p > r)dr\\
&=\int_{0}^{\infty}\Pr(\|U_L\|_p > r^{1/p})dr. 
\label{eq:nonegativeexpectation}
\end{align}
We further have that 
\begin{align}
\Pr(\|U_L\|_p> r^{1/p})=1-\Pr(\|U_L\|_p\leq r^{1/p}),    
\end{align}
and that
\begin{align}
\Pr(\|U_L\|_p\leq r^{1/p})=|r^{1/p}\mBp\cap \mVLp|=m_L(\pi_L(r^{1/p}\mBp)),    
\end{align}
where the justification of the last equality is similar to that of~\eqref{eq:RprojgC}.

We have therefore obtained that 
\begin{align}
\EE\|U_L\|^p_p=\int_{0}^{\infty}1-m_L(\pi_L(r^{1/p}\mBp)) dr.
\end{align}
Now, further averaging with respect to $L\sim\mu_n$ and using Tonelli's Theorem, we get
\begin{align}
\EE\|U_L\|^p_p=\int_{0}^{\infty}1-\EE[m_L(\pi_L(r^{1/p}\mBp))] dr.
\end{align}
Using~\eqref{eq:Eml_lb} applied with $\KK=r^{1/p}\mBp$, we have
that
\begin{align}
\EE[m_L(\pi_L(r^{1/p}\mBp))]\geq \frac{|r^{1/p}\mBp|}{1+|r^{1/p}\mBp|} =\frac{r^{n/p}\Vp}{1+r^{n/p}\Vp},   
\end{align}
and therefore
\begin{align}
\EE\|U_L\|^p_p&\leq \int_{0}^{\infty}\frac{1}{1+(r\Vp^{p/n})^{n/p}} dr\\
&=\frac{1}{\Vp^{p/n}}\int_{0}^{\infty}\frac{1}{1+t^{n/p}} dt.
\end{align}
Finally, using~\eqref{eq:sincintegral} we have that for any $\nu>0$,
\begin{align}
\int_{0}^{\infty}\frac{1}{1+t^\nu} dt=\frac{\pi/\nu}{\sin(\pi/\nu)}=\frac{1}{\sinc(1/\nu)},
\end{align}
and we obtain that
\begin{align}
\EE\|U_L\|^p_p&\leq \frac{1}{\Vp^{p/n}}\cdot\frac{1}{\sinc( p/n)},
\end{align}
establishing the claimed result.
\end{proof}

To put this result in context, it is easy to see that for any unit covolume lattice $L\subset\R^n$ it holds that
\begin{align}
\GPl\geq G^*_{n,p}\df \frac{1}{(n+p)\Vp^{p/n}}.    
\end{align}
This follows by setting $S=r_{\text{eff}}\mBp$, where $r_{\text{eff}}=\frac{1}{\Vp^{1/n}}$ is chosen so that $|S|=1=|\mVLp|$, and computing $\EE\|U_{S}\|_p^p$ for $U_S\sim\Unif(S)$. Clearly, $\EE\|U_{L}\|_p^p\geq \EE\|U_S\|_p^p$, and $\EE\|U_{S}\|_p^p$ can be computed as in~\eqref{eq:nonegativeexpectation}, noting that
\begin{align}
  \Pr(\|U_S\|^p_p\geq r)=\begin{cases}
  1-r^{n/p}\Vp & 0\leq r\leq r_{\text{eff}}^p\\
  0 & r>r_{\text{eff}}^p
  \end{cases}.
\end{align}
The following is the analogue of Lemma~\ref{lem:Markov2} for general $p>0$.

\begin{lem}
\label{cor:NpMconcentration}
Let  $p>0$ and $n\geq 4 p$ be an integer. Then, for $L\sim\mu_n$ 
 \begin{align}
 \Pr\left(\GPl>(1+\kappa)G^*_{n,p}\right)\leq \frac{2p}{\kappa n}.
 \end{align}
\end{lem}

\begin{proof}
Define the random variable $Y=\frac{G^{(p)}_L}{G_{n,p}^*}-1$, which is non-negative with probability $1$.
By Markov's inequality, we have    
\begin{align}
\Pr\left(G^{(p)}_L\geq(1+\kappa) G^*_{{n,p}} \right)=\Pr(Y\geq \kappa)\leq\frac{\EE[Y]}{\kappa}.    
\end{align}
Using $\sin(x)\geq x(1-\frac{x^2}{6})$ for $0<x<1$, we have that for any $0<x\leq \frac{1}{\pi}$ it holds that
\begin{align}
\frac{1}{\sinc(x)}\leq\frac{\pi x}{\pi x(1-\frac{(\pi x)^2}{6})}\leq 1+\frac{1}{5}(\pi x)^2<1+2 x^2.   
\end{align}
Applying Theorem~\ref{thm:excpectedNpMbound} we obtain (for $n\geq 4 p$ such that $\frac{p}{n}<\frac{1}{\pi}$)
\begin{align}
\EE[Y]&\leq \frac{n+p}{n}\frac{1}{\sinc(p/n)}-1\\
&\leq \left(1+\frac{p}{n}\right)\left(1+2\left(\frac{p}{n}\right)^2\right) -1\\
&=\frac{p}{n}+2\frac{p^2}{n^2}+2\frac{p^3}{n^3}<2\frac{p}{n},
\end{align}
which yields the claimed result.
\end{proof}

\section{Proofs of claims for linear codes}
\label{app:discreteCDFproperties}

\begin{proof}[Proof of Proposition~\ref{prop:LC_fundviacdf}]
The first item follows since 
\begin{align}
\rpack(\CC)&=\max\{r~:~\BB_{n,r}\subset\VC\}\\
&=\max\{r~:~Q_{\CC}(r)=2^{k-n} V_{n,r}\},
\end{align}
where in the second equality we have used~\eqref{eq:QasIntersecVol}. The second item follows since
\begin{align}
 \rcov(\CC)=\max_{x\in\F_2^n}\min_{y\in\CC}d_H(x,y)=\max_{x\in \VC}|x|.   
\end{align}
For the third item, we write
\begin{align}
\EE|U_L|&=\sum_{r=0}^n \Pr(|U_\CC|> r)\\
&= \sum_{r=0}^n 1-\Pr(|U_\CC|\leq r)\\
&=\sum_{r=0}^n 1-Q_\CC(r).   
\end{align}
The fourth item follows directly from Proposition~\ref{prop:ProbAsCQFexpt_dicrtete}.
\end{proof}

% \begin{proof}[Proof of Proposition~\ref{prop:ProbAsCQFexpt_dicrteteOld}]
% We have
% \begin{align}
% \Pr(Z\in \A)&=\sum_{r=0}^n \Pr\left(Z\in (\A\cap\SSS_{n,r}\right)\\
% &=\sum_{r=0}^n p^{r}(1-p)^{n-r} \left|\A\cap\SSS_{n,r}\right|\\
% &=\sum_{r=0}^n {{n}\choose{r}} p^{r}(1-p)^{n-r} \frac{\left|\A\cap\SSS_{n,r}\right|}{{{n}\choose{r}}}\\
% &=\sum_{r=0}^n \Pr(|Z|=r) \frac{\left[\left|\A\cap\BB_{n,r}\right|-\left|\A\cap\BB_{n,r-1}\right|\right]}{{{n}\choose{r}}}\\
% &=\sum_{r=0}^n \Pr(|Z|=r) \frac{|\A|\cdot \left[Q_\A(r)-Q_\A(r-1)\right]}{{{n}\choose{r}}}\\
% &=|\A|\cdot\EE\left[\frac{\left(Q_\A(|Z|)-Q_\A(|Z|-1)\right)}{{{n}\choose{|Z|}}} \right],
% \end{align}
% as claimed.
% \end{proof}

\begin{proof}[Proof of Proposition~\ref{prop:ProbAsCQFexpt_dicrtete}]
Let $\varphi_p(r)=p^r(1-p)^{n-r}$, and not that $\varphi_p(r)-\varphi_{p}(r+1)=\varphi_p(r)\cdot\frac{1-2p}{1-p}$.
We have
\begin{align}
&\Pr(Z\in \KK)=\sum_{r=0}^n \Pr\left(Z\in (\KK\cap\SSS_{n,r})\right)\\
&=\sum_{r=0}^n \varphi_p(r) \left|\KK\cap\SSS_{n,r}\right|\\
&=\sum_{r=0}^n \varphi_p(r) \left(\left|\KK\cap\BB_{n,r}\right|-\left|\KK\cap\BB_{n,r-1}\right|\right)\\
&=\varphi_p(n)\left|\KK\cap\BB_{n,n}\right|+\sum_{r=0}^{n-1} (\varphi_p(r)-\varphi_p(r+1)) \left|\KK\cap\BB_{n,r}\right|\\
&=\left(\frac{p}{1-p}+\frac{1-2p}{1-p}\right)\varphi_p(n)\left|\KK\cap\BB_{n,n}\right|\nonumber\\
&+\frac{1-2p}{1-p}\sum_{r=0}^{n-1} \varphi_p(r)\left|\KK\cap\BB_{n,r}\right|\\
&=|\KK|\cdot\frac{p}{1-p}\cdot p^n+\frac{1-2p}{1-p}\sum_{r=0}^{n} \varphi_p(r)\left|\KK\cap\BB_{n,r}\right|\\
&=|\KK|\cdot\frac{p}{1-p}\cdot p^n\nonumber\\
&+|\KK|\cdot\frac{1-2p}{1-p}\sum_{r=0}^{n} {{n}\choose{r}} p^{r}(1-p)^{n-r}\frac{\left|\KK\cap\BB_{n,r}\right|/|\KK|}{{{n}\choose{r}}}\\
&=|\KK|\left(\frac{p}{1-p}\cdot p^n+ \frac{1-2p}{1-p}\sum_{r=0}^{n} {{n}\choose{r}} p^{r}(1-p)^{n-r}\frac{Q_\KK(r)}{{{n}\choose{r}}}\right)\\
&=|\KK|\left(\frac{p}{1-p}\cdot p^n+ \frac{1-2p}{1-p}\EE\left[\frac{Q_\KK(|Z|)}{{{n}\choose{|Z|}}}\right]\right),
\end{align}
as claimed.
\end{proof}

\begin{proof}[Proof of Theorem~\ref{thm:DcConstantGap}]
Let $m<n/2-1$ be a positive integer, and denote
\begin{align}
\beta_m=\frac{n-m}{m+1}>1.    
\end{align}
We claim that 
\begin{align}
\gamma_r=\frac{{{n}\choose{r+1}}}{V_{n,r}}\geq \beta_m-1,~~~ \forall 0\leq r\leq m.   
\label{eq:gammainduc}
\end{align}
We show this by induction. For $r=0$, we have $\gamma_0=n>\beta_m-1$. Assume $\gamma_\ell>\beta_m-1$ for all $0\leq \ell\leq r-1<m$. We have that
\begin{align}
\gamma_{r}&=\frac{{{n}\choose{r+1}}}{V_{n,r}}=\frac{{{n}\choose{r+1}}}{{{n}\choose{r}}+V_{n,r-1}}=\frac{{{n}\choose{r+1}}}{{{n}\choose{r}}\left(1+\frac{1}{\gamma_{r-1}}\right)}\nonumber\\
&=\frac{n-r}{r+1}\frac{\gamma_{r-1}}{1+\gamma_{r-1}}\geq \beta_m\cdot \frac{\gamma_{r-1}}{1+\gamma_{r-1}}\geq 1-\beta_m,   
\end{align}
where the last inequality follows since $t\mapsto\frac{t}{1+t}$ is monotonically increasing, and the induction assumption that $\gamma_{r-1}>1-\beta_m$. 

Using~\eqref{eq:gammainduc}, we have that
\begin{align}
\frac{x_{r+1}}{x_r}&=\frac{V_{n,r+1}}{V_{n,r}}\\
&=1+\frac{{{n}\choose{r+1}}}{V_{n,r}}\\
&=1+\gamma_r\\
&\geq \beta_m,~~~  \forall 0<r\leq m.
\label{eq:xrratiobound}
\end{align}

By~\eqref{eq:Deltank}, for any $\reff(n,k)<m<n/2-1$ we can write
\begin{align}
\Delta(n,k)&\df S_1+S_{2a}+S_{2b},
\end{align}
where
\begin{align}
 S_1&=\sum_{r=0}^{\reff(n,k)-1}\frac{x_r^2}{1+x_r}\leq \sum_{r=0}^{\reff(n,k)-1} x_r\\
 S_{2a}&=\sum_{r=\reff(n,k)}^m \frac{1}{1+x_r}\leq \sum_{r=\reff(n,k)}^m x_r^{-1}\\
 S_{2b}&=\sum_{r=1+m}^n \frac{1}{1+x_r}\leq n \cdot x^{-1}_{m}.
\end{align}
Note that by definition of $\reff(n,k)$, we have that $x_r< 1$ for all $r\leq \reff(n,k)-1$, and that $x_r\geq 1$ for all $r\geq  \reff(n,k)$. By~\eqref{eq:xrratiobound} we therefore have
\begin{align}
x_r&\leq \beta_m^{-(\reff(n,k)-1-r)},~~~\forall 0\leq r\leq \reff(n,k)-1,\\
x^{-1}_r&\leq \beta_m^{-(r-\reff(n,k))},~~~\forall  \reff(n,k)\leq r\leq m.
\end{align}
It therefore follows that
\begin{align}
S_1&\leq\sum_{r=0}^{\reff(n,k)-1} x_r\leq \sum_{r=0}^{\reff(n,k)-1} \beta_m^{-(\reff(n,k)-r)}<\sum_{j=0}^\infty\beta_m^{-j}\nonumber\\
&=\frac{1}{1-\beta_m^{-1}}\\
S_{2a}&\leq\sum_{r=\reff(n,k)}^m x^{-1}_r\leq \sum_{r=\reff(n,k)}^m \beta_m^{-(r-\reff(n,k))}<\sum_{j=1}^\infty\beta_m^{-j}\nonumber\\
&\leq\frac{1}{1-\beta_m^{-1}}.
\end{align}
Thus,
\begin{align}
    \Delta(n,k)\leq 2\frac{\beta_m}{\beta_m-1}+n\cdot  \beta_m^{-(m-\reff(n,k))}.
    \label{eq:DeltaInter}
\end{align}
Take $m=\left\lceil \frac{n}{2}\left(h_2^{-1}(1-\alpha)+\frac{1}{2} \right) \right\rceil$, where $h_2(p)=-p\log_2(p)-(1-p)\log_2(1-p)$, and $h_2^{-1}(\cdot)$ is its inverse restricted to $[0,1/2)$. We may assume without loss of generality that $n$ is large enough such that $\kappa\df\frac{m+1}{n}<1/2$, so that $\beta_m\geq \frac{1-\kappa}{\kappa}>1$, and $m-\reff(n,k)>\left\lceil\frac{\log n}{\log \beta_m}\right\rceil$. Substituting this into~\eqref{eq:DeltaInter}, we obtain
\begin{align}
\Delta(n,k)\leq 2\frac{1-\kappa}{1-2\kappa}+1,    
\end{align}
as claimed.
\end{proof}

\begin{proof}[Proof of Lemma~\ref{lem:sphereDistortion}]
Let $X\sim\mathrm{Ber}^{\otimes n}(1/2)$ and let $\CC=\{c_1,\ldots,c_{2^k}\}\subset\{0,1\}^n$ be the codebook corresponding to the image of the decoder $g(\cdot)$. Define the random variable
\begin{align}
Y=d_H(X,\CC)=\min_{c\in\CC}d_H(X,c).    
\end{align}
Since given $\CC$ the optimal encoder $f$ maps each $X$ to the nearest point in $\CC$, we have that the average distortion of the code is at least
\begin{align}
D=\EE[Y]=\sum_{y=0}^n\Pr(Y>y)=\sum_{y=0}^n 1-\Pr(Y\leq y). 
\label{eq:DeqGenCode}
\end{align}
Let
\begin{align}
\A_y=\{x\in\F_2^n~:~d(x,\CC)\leq y\}.   
\end{align}
We clearly have that $\A_y\subset \cup_{c\in\CC}\{c+B_{n,y}\}$, and therefore $|\A_y|\leq \min\{|\CC|\cdot V_{n,y},2^n\}$, and consequently,
\begin{align}
\Pr(Y\leq y)=\Pr(X\in \A_y)=2^{-n}|\A_y|\leq \min\{2^{k-n}V_{n,y},1\}.   \end{align}
Substituting this into~\eqref{eq:DeqGenCode}, and recalling~\eqref{eq:DnkDef}, yields the claimed result.
\end{proof}

\bibliographystyle{IEEEtran}
\bibliography{elonbib}

\end{document}